%% file: genmala4.tex
\documentclass[12pt]{article}
\usepackage[utf8]{inputenc}
\usepackage[english]{babel}

\usepackage{times}
\usepackage[margin=1in]{geometry}
\usepackage{natbib,amsthm,amsmath,amsfonts,amssymb}
\usepackage{amsbsy,mathrsfs,subfigure,graphicx}
\usepackage{bm}
\usepackage{caption}
\usepackage{multirow}
\usepackage{mathtools}
\usepackage{algorithm}
\usepackage{setspace}

\newtheorem{theorem}{Theorem}

\newtheorem{corollary}{Corollary}
\newtheorem{proposition}{Proposition}
\theoremstyle{definition}

\newtheorem{remark}{Remark}

\DeclareMathOperator{\Cov}{Cov}

\newcommand{\TV}{{\mathrm{TV}}}
\newcommand{\E}{{\mathrm{E}}}

\newcommand{\pcite}[1]{\citeauthor{#1}'s \citeyearpar{#1}}
 
\newcommand{\sX}{\textsf{X}}
\newcommand{\Sbb}{{\mathbb{S}}}

\graphicspath{{./plots/}}

\allowdisplaybreaks

\begin{document}

\title{Convergence of position-dependent MALA with application to conditional simulation in GLMMs}

\author{Vivekananda Roy and
  Lijin Zhang\\
 Department of Statistics, Iowa State University, USA}

\date{}

\maketitle              

\begin{abstract}
  We establish conditions under which Metropolis-Hastings
  (MH) algorithms with a position-dependent proposal covariance matrix
  will or will not have the geometric rate of convergence. Some of the
  diffusions based MH algorithms like the Metropolis adjusted Langevin
  algorithm (MALA) and the pre-conditioned MALA (PCMALA) have a
  position-independent proposal variance. Whereas, for other modern
  variants of MALA like the manifold MALA (MMALA) that adapt to the
  geometry of the target distributions, the proposal covariance matrix
  changes in every iteration. Thus, we provide conditions for
  geometric ergodicity of different variations of the Langevin
  algorithms. These results have important practical implications as
  these provide crucial justification for the use of asymptotically
  valid Monte Carlo standard errors for Markov chain based
  estimates. The general conditions are verified in the context of
  conditional simulation from the two most popular generalized linear
  mixed models (GLMMs), namely the binomial GLMM with the logit link
  and the Poisson GLMM with the log link. Empirical comparison in the
  framework of some spatial GLMMs shows that the computationally less
  expensive PCMALA with an appropriately chosen pre-conditioning
  matrix may outperform the MMALA.
\end{abstract}




\section{Introduction}
\label{sec:int}
In physics, statistics, and several other disciplines one often deals
with a complex probability density $f(x)$ on $\mathbb{R}^d$ that is
available only up to a normalizing constant. Generally, the goal is to
estimate $\E_f [g] :=\int_{\mathbb{R}^d} g(x) f(x) dx$ for some real valued
function $g$. Markov chain Monte Carlo (MCMC) is the most popular
method for sampling from such a $f$ and for providing a Monte Carlo
estimate of $E_f [g]$ \citep{robe:case:2004}.  In MCMC, a Markov chain
$\{X_n\}$, which has $f$ as its stationary density, is run for
a certain number of iterations, and $\E_f [g]$ is estimated by the sample
average $\bar{g}_n := \sum_{i=1}^n g(X_i)/n$. Among the different
MCMC algorithms, Metropolis-Hastings (MH) algorithms
\citep{metr:nich:1953, hast:1970} are predominant. In MH
algorithms, given the current state $x$, a proposal $y$ is drawn from
a density $q(x, y)$, which is then accepted with a certain probability. The accept-reject step guarantees reversibility of the Markov
chain with respect to the target $f$, and, in turn, ensures
stationarity.  Besides, the acceptance probability generally does not
involve the unknown normalizing constant in $f$, making the implementation of
MH algorithms practically feasible.

A popular MH algorithm is the random walk Metropolis (RWM) where the
proposal density is $N (x, hI_d)$, the normal density centered at the
current state $x$ and with the covariance matrix $hI_d$ for some
$h>0$. A nice feature of the RWM is that the acceptance probability
can be adjusted by choosing the step-size (proposal variance) $h$
accordingly. Indeed, lower step-size results in a higher acceptance
probability but then the RWM chain takes longer to move around the
space. Therefore, in higher dimensions, that is when $d$ is large, the
Metropolis adjusted Langevin algorithm (MALA)
\citep{ross:doll:frie:1978, besa:1994, robe:twee:1996b}, which employs
the gradient of log of the target distribution, is developed to
achieve faster mixing. Since the mean of the proposal density
$N(x+h\nabla\log f(x)/2, hI_d)$ of the MALA is governed by the
gradient information, it is likely to make moves in the directions in
which $f$ is increasing. Thus, large proposals can be accepted with a
higher probability, leading to high mixing of the Markov chain. On the
other hand, the proposal density in the RWM does not make use of the
structure of the target density. Superiority of the MALA over the RWM
in terms of mixing time is demonstrated by \cite{robe:rose:1998}
\cite[see also][]{chri:robe:rose:2005, dwiv:chen:wain:yu:2019,
  chen:dwiv:wain:yu:2020,lee:shen:2020,wu:schm:2021}.

However, MALA may be inefficient when the coordinates of $x$ are
highly correlated, and have largely differing marginal variances. In such situations, the step-size is compromised to accommodate the
coordinates with the smallest variance. Such a situation arises when
modeling spatially correlated data. Spatial models take the
correlation of different locations into consideration, usually, the
closer the two locations, the more similarity and the higher correlations they have.
The pre-conditioned MALA (PCMALA) \citep{robe:stra:2007} is introduced to
circumvent these issues by multiplying a covariance matrix to the
gradient of log of the target density. The proposal density
of the PCMALA is $N(x+hG\nabla\log f(x)/2, hG)$, while the selection of
an appropriate covariance matrix requires further study. Without the Metropolis step,
the MALA and the PCMALA degenerate to the unadjusted Langevin algorithm (ULA)
\citep{pari:gior:1981, gren:mill:1994, robe:twee:1996b} and the pre-conditioned ULA (PCULA), respectively, which although
might converge to undesired distributions, require less
computational time. 

By taking into account the geometry of the target distribution in the
selection of step-sizes, efficient versions of MALA can be formed that
adapt to the characteristics of the target. Indeed, using ideas from
both Riemannian and information geometry, \cite{giro:cald:2011}
propose a generalization of MALA, called the manifold MALA
(MMALA). MMALA is constructed taking into account the natural geometry
of the target density and considering a Langevin diffusion on a
Riemann manifold. In the MMALA, the covariance matrix $G$, unlike the
PCMALA, changes in every iteration. More recently,
\cite{xifa:sher:2014} propose the position-dependent MALA
(PMALA). There are other works in the literature \cite[see
e.g][]{haar:saks:tamm:2001, robe:rose:2009}, which consider RWM
algorithms where the Gaussian proposal distribution is centered at the
current state and the covariance matrix depends on the current or a
finite number of previous states.

A Harris ergodic Markov chain will converge to the target
distribution, and $\E_f [g]$ can be consistently estimated by the
sample mean $\bar{g}_n$ \citep{meyn:twee:1993}. On the other hand, in
practice, it is important to ascertain the errors associated with the
estimate $\bar{g}_n$. Establishing geometric ergodicity of a Markov
chain is the most standard method for guaranteeing a central limit
theorem (CLT) for $\bar{g}_n$ and finding its standard errors. Thus,
the geometric rate of convergence for Markov chains is highly
desired. Furthermore, a non-geometrically ergodic chain may sample
heavily from the tails instead of the center of the distribution,
leading to instability of the Monte Carlo estimation
\citep{robe:twee:1996}.  The main contribution of this article is that
it establishes conditions under which geometric ergodicity will and
will not hold for the position-dependent MALA. Indeed, we provide
these results for MH algorithms where the normal proposal density has
a general mean function $c(x)$ and a covariance matrix $G(x)$
depending on the current position $x$ of the Markov chain. As special
cases, these results also hold for the MMALA, the PMALA, and the
PCMALA. We also provide conditions guaranteeing geometric ergodicity
of the PCULA. Our results will help practitioners implementing these
MCMC algorithms to choose appropriate step-sizes ensuring the
geometric convergence rates. Previously, \cite{robe:twee:1996} derive
conditions under which the MALA and ULA chains are geometrically
ergodic (GE). However, in the literature, there is no results available on
convergence analysis of position dependent MALA chains. Recently,
\cite{livi:2021} considers ergodicity properties of the RWM algorithm
with a position-dependent proposal variance. Some of these previously
mentioned results are valid only for $d=1$. It is known that the
Hamiltonian Monte Carlo (HMC) algorithm with exactly one leapfrog step
boils down to the MALA \citep{neal:2011}. \cite{livi:beta:byrn:2019}
establish geometric ergodicity of the HMC when the `mass matrix' in
the `kinetic energy' is a fixed matrix. On the other hand, in our
geometric convergence results, the pre-conditioning covariance matrix
is allowed to vary with the current position of the Markov chains.

Generalized linear mixed models (GLMMs) are often used for analyzing
correlated non-Gaussian data. Spatial generalized linear mixed models
(SGLMMs) are GLMMs where the correlated random effects form the
underlying Gaussian random fields. SGLMMs are useful for modeling
spatially correlated binomial and count data. Simulation from the
random effects given the observations from a GLMM or a SGLMM is
important for prediction and the Monte Carlo maximum likelihood
estimation \citep{digg:tawn:moye:1998, geye:1994a}. Langevin
algorithms have been previously used for making inference in the
SGLMMs \citep{chri:moll:waag:2001, chri:robe:rose:2005,
  chri:robe:skol:2006}. Another contribution of this paper is to
verify our general conditions for geometric ergodicity of different
versions of the MALA and the ULA for conditional simulation in the GLMMs. In
particular, using our general sufficient conditions mentioned before,
we establish the geometric rate of convergence of different versions
of the MALA with appropriately chosen step-sizes for the binomial GLMM. On
the other hand, our general necessary conditions are used to show that
the PCMALA is {\it not} geometrically ergodic for the Poisson GLMM. We
also undertake empirical comparisons of the before mentioned
algorithms in the context of simulated data from high dimensional
SGLMMs. In the numerical examples we observe that the PCMALA compares
favorably with the computationally expensive PMALA. Avoiding expensive
computation of derivatives repeatedly in each iteration, the PCMALA is
computationally efficient. On the other hand, computational cost for
the MMALA and other MCMC algorithms with a position-dependent proposal
variance may not scale favorably with increasing dimensions as noted
in \cite{giro:cald:2011}. \cite{giro:cald:2011} compare the MMALA with
the `simplified MMALA' where the metric tensor is a locally constant
in the context of several examples and they observe that, although the
simplified MMALA is `computationally much less expensive', it is less
efficient.  \cite{giro:cald:2011} argue that `a global level
of pre-conditioning may be inappropriate for differing transient and
stationary regimes', however, we observe that for the SGLMM examples considered
here, the PCMALA with a well-chosen (suggested by
\cite{giro:cald:2011} themselves) pre-conditioning matrix can
outperform the PMALA and the MMALA for chains started either at the
center or away from the mode.

The rest of the paper is organized as follows. Section~\ref{sec:gmala}
contains a brief review of the MALA and its different variants. After
discussing some basic results on convergence of Markov chains in
Section~\ref{sec:mcbg}, we provide our main results on MH algorithms
with a position-dependent proposal variance in
Section~\ref{sec:gmalage}.  Section~\ref{sec:gegula} contains some
convergence results for the PCULA. Our general convergence results for
variations of the MALA are demonstrated for GLMMs in
Section~\ref{sec:sglmm}. This section also contains empirical
comparisons between different variants of the Langevin algorithms in the context
of conditional simulation for the SGLMMs. Some concluding remarks appear in 
Section~\ref{sec:disc}. Finally, most of the proofs and some numerical results are given in the Supplement.
The sections in the supplement document are referenced here with the prefix `S'.

\section{Metropolis adjusted Langevin algorithms}
\label{sec:gmala}
MALA is a discrete time MH Markov chain
based on the Langevin diffusion $X_t$ defined as
\begin{equation}
  \label{eq:lang}
    dX_t=(1/2)\nabla\log f(X_t)dt+dW_t,  
\end{equation}
where $W_t$ is the $d-$dimensional standard Brownian motion. Although
$f$ is stationary for $X_t$ in \eqref{eq:lang}, simple
discretizations of it say,
$X_n = X_{n-1}+ h\nabla\log f(X_{n-1})/2 + \sqrt{h} \epsilon_n$ for a
chosen step-size $h$ with $\epsilon_n \stackrel{iid}{\sim} N(0, I_d)$
can fail to maintain the stationarity. This is why, in the MALA, an MH
accept-reject step is introduced where, in each iteration, the
proposal $X_n$ drawn from
$N(x_{n-1}+ h\nabla\log f(x_{n-1})/2, h I_d)$ is only accepted with
probability
\begin{equation}
  \label{eq:accpro}
      \alpha(X_{n-1},X_n)=1\wedge\frac{f(X_n)q(X_{n},X_{n-1})}{f(X_{n-1})q(X_{n-1},X_{n})},
\end{equation}
where the proposal
density $q(x,y)$ is the $N(x+h\nabla\log f(x)/2, hI_d)$ density evaluated at $y$. Several extensions of the MALA have been proposed in
the literature. These variants are based on different stochastic
differential equations $d X_t = b(X_t) dt + \sigma(X_t) dW_t$
with a certain drift vector $b(x)$ and a volatility matrix $\sigma(x)$. The Fokker-Planck equation given by 
\begin{equation}
  \label{eq:fokk}
  \frac{\partial}{\partial t} u(x, t) = -\sum_{i} \frac{\partial}{\partial x_i} [b_i (x) u(x, t)] + \sum_{ij} \frac{\partial^2}{\partial x_i \partial x_j} [D_{ij} (x) u(x, t)]
\end{equation}
describes the evolution of the pdf $u(x, t)$
of $X_t$. Here, $D(x) = \sigma(x) \sigma(x)^{\top}/2$ is the diffusion
coefficient. If $u(x, t) = f(x) \;\forall t$ then the process
$\{ X_t\}_{t \ge 0}$ is stationary with the invariant density $f$. Setting $u(x, t) = f(x), b(x)= \nabla\log f(x)/2$ and $\sigma(x) = I$, it can be seen that \eqref{eq:fokk} holds for \eqref{eq:lang}. A generalization of \eqref{eq:lang} still satisfying \eqref{eq:fokk} is the diffusion
\begin{equation}
  \label{eq:pclang}
    dX_t=(1/2)G\nabla\log f(X_t)dt+\sqrt{G} dW_t,  
\end{equation}
for a positive definite matrix $G$. The corresponding discrete time MH
chain with proposal density $N(x+hG\nabla\log f(x)/2, hG)$ is known as
the pre-conditioned MALA (PCMALA) \citep{robe:stra:2002}. By choosing $G$
appropriately, the PCMALA can be well suited to situations where
coordinates of the random vector following $f$ are highly
correlated, and have different marginal variances. In
Section~\ref{sec:sglmm}, we discuss several choices of $G$.

In both \eqref{eq:lang} and \eqref{eq:pclang} the volatility matrix is
constant. \cite{giro:cald:2011} and \cite{xifa:sher:2014} propose
variants of \eqref{eq:lang} with a position-dependent volatility
matrix. The MH proposal of \pcite{xifa:sher:2014} position-dependent MALA (PMALA) is driven by
\begin{equation}
  \label{eq:plang}
    dX_t=(1/2)G(X_t)\nabla\log f(X_t)dt+(1/2) \Gamma(X_t) dt + \sqrt{G(X_t)} dW_t,  
\end{equation}
where
$\Gamma_i(X_t) = \sum_j \partial G_{ij} (X_t)/\partial
X_{t,j}$. Straightforward calculations show that \eqref{eq:plang}
satisfies \eqref{eq:fokk}. In practice, we often use
$G(X_t) = \mathbb{I}^{-1}(X_t)$ for some appropriate choice of $\mathbb{I}$. In that
case,
$\Gamma_i(X_t) = \sum_j \partial \mathbb{I}^{-1}_{ij} (X_t)/\partial
X_{t,j}$.  The proposal transition of \pcite{giro:cald:2011} manifold
MALA (MMALA) is driven by a  diffusion on a Riemannian manifold given by
\begin{equation}
  \label{eq:mlang}
    dX_t=(1/2)\mathbb{I}^{-1}(X_t)\nabla\log f^*(X_t)dt+(1/2) \Omega(X_t) dt + \sqrt{\mathbb{I}^{-1}(X_t)} dW_t,  
\end{equation}
with $f(x) = f^*(x) |\mathbb{I}(x)|^{1/2}$ and
$\Omega_i = \sum_j \partial \mathbb{I}^{-1}_{ij}/\partial X_{t,j} + 0.5 \sum_j
\mathbb{I}^{-1}_{ij} \partial \log |\mathbb{I}|/\partial X_{t,j}$. Here, we have
accounted for a transcription error of \cite{giro:cald:2011} as
mentioned in \cite{xifa:sher:2014}. From
\eqref{eq:mlang}, it follows that the proposal density of the MMALA
chain is
$N(x+(h/2)\mathbb{I}^{-1}(x)\nabla\log f^*(x)+(h/2) \Omega(x),
h\mathbb{I}^{-1}(x))$. In this article, we study
convergence properties of MH algorithms with the candidate distribution
$N(c(x), h G(x))$ for some general mean vector $c(x)$ and the covariance
matrix $hG(x)$. This will cover as special cases different variants of
the MALA discussed before.

\section{Markov chain background}
\label{sec:mcbg}
Let $(\sX, \mathcal{B})$ denote a Borel space. Here, we consider
$\sX = \mathbb{R}^d$ and let $\|\cdot \|$ denote the Euclidean
norm. Let $F$ denote the target probability measure and
$P(x, dy): \sX \times \mathcal{B} \rightarrow [0,1]$ be a Markov
transition function (Mtf). 
We will use $f(x)$ to denote the pdf of $F$ with respect to the
Lebesgue measure. Let $\{X_n\}_{n=0}^{\infty}$ be a Markov chain
driven by $P$. Let $P^n(\cdot,\cdot)$ denotes the $n-$step Mtf. Now,
$P$ is $\phi$-irreducible if there exists a non-zero $\sigma$-finite measure
$\phi$ on $\sX$ such that for all $A \in \mathcal{B}$ with $\phi(A) > 0$, and for all $x \in \sX$, there
exists a positive integer $n = n(x, A)$ such that $P^n(x, A) > 0$.
If $P$ is $\phi$-irreducible and $F$ is invariant with
respect to $P$, then $\{X_n\}_{n=0}^{\infty}$ can be used to
consistently estimate means with respect to $f$ \citep[][Chap
10]{meyn:twee:1993}. Indeed, under these conditions, if $g:\sX \rightarrow \mathbb{R}$ is
integrable with respect to $F$, that is, if
$\E_{f} [|g(x)|] := \int_{\sX} |g(x)| f(x)dx < \infty$, then
$\overline{g}_n := \sum_{i=0}^{n-1} g(X_i)/n \rightarrow \E_{f} [g]$
almost surely, as $n \rightarrow \infty$. On the other hand, Harris ergodicity of $P$ does not guarantee
a CLT for $\overline{g}_n$. We say a CLT for $\overline{g}_n$ exists
if
$\sqrt{n}( \overline{g}_n - \E_{f}[g]) \stackrel{d}{\rightarrow}
N(0,\sigma_{g}^2)$ as $n \rightarrow \infty$ for some
$\sigma^2_{g} \in (0,\infty)$. The most common method for ensuring
a Markov chain CLT is to establish that $\{X_n\}_{n=0}^{\infty}$ ($P$) is
geometrically ergodic (GE), that is, to demonstrate the existence of a function
$L:\sX \rightarrow [0,\infty)$ and a constant $\rho \in (0,1)$, such that
for all $n=0,1,2,\dots$,
\begin{align}
\label{eq:ge}
 \| P^n(x,\cdot)- F(\cdot)\|_{\TV} \le L(x)\rho^n,  \;\; x \in \sX,
\end{align}
where
$\| \cdot\|_{\TV}$ denotes the total variation norm. \eqref{eq:ge}
guarantees a CLT for $\overline{g}_n $ if
$\E_{f}[ g^{2+\delta}] < \infty$ for some $\delta >0$. \eqref{eq:ge}
also implies that a valid standard error $\hat{\sigma}_{g}/\sqrt{n}$
for $\overline{g}_n$ can be calculated by the batch means or the
spectral variance methods, which, in turn, can be used to decide `when
to stop' running the Markov chain \citep{vats:fleg:jone:2019,
  roy:2020}. Furthermore, as mentioned in \cite{roy:2020}, most of
the MCMC convergence diagnostics used in practice, for example, the
effective sample size and the potential scale reduction factor used
later in this paper, assume the existence of a Markov chain CLT,
emphasizing the importance of establishing \eqref{eq:ge}.

If $P$ is $\phi-$irreducible and aperiodic, then from
\pcite{meyn:twee:1993} chap 15, we know that \eqref{eq:ge} is equivalent to the
existence of a Lyapunov function $V: \sX \rightarrow [1, \infty]$ and constants
$\lambda < 1, b < \infty$ with
\begin{equation}
  \label{eq:drif}
  PV(x) \le \lambda V(x) + b I_C(x), \;\; x \in \sX,
\end{equation}
where $PV(x) = \int_{\sX} V(y) P(x, dy) = \E[V(X_1)|X_0=x]$ and
$C \subset \sX$ is small, meaning that $\exists \; \varepsilon>0$, integer
$k$, and a probability measure $\nu$ such that
$P^k(x, A) \ge \varepsilon \nu(A)$ $\forall x \in C$, and
$A \in \mathcal{B}$.

In the presence of some topological properties, we can use the
following result to establish geometric ergodicity of a Markov chain. The
function $V: \sX \rightarrow [0, \infty]$ is said to be unbounded off
compact sets if for any $a>0$, the level set
$\{x \in \sX: V(x) \le a\}$ is compact. The next proposition, which directly
follows from several results in \citet{meyn:twee:1993}, has been used for establishing
geometric ergodicity of different MCMC algorithms \cite[see e.g.][]{roy:hobe:2007,wang:roy:2018b}.
\begin{proposition}[Meyn and Tweedie]
  \label{prop:gedrif}
  Let $P$ be $\phi$-irreducible, aperiodic and Feller, where $\phi$ has nonempty interior. Suppose $V: \sX \rightarrow [0, \infty]$ is unbounded off compact sets such that 
\begin{equation}
  \label{eq:drif2}
  PV(x) \le \lambda V(x) + b,
\end{equation}
for all $x$ and for some constants
$\lambda < 1, b < \infty$, then $\{X_n\}_{n=0}^{\infty}$ is
GE.
\end{proposition}
\begin{proof}[Proof of Proposition~\ref{prop:gedrif}]
Let $V'(x) = V(x)+ 1$. Then $V': \sX \rightarrow [1, \infty]$ is also unbounded off compact sets and \eqref{eq:drif2} holds for $V'$ with $b$ replaced by $b+1-\lambda$. By \citet[][Theorem 6.0.1]{meyn:twee:1993} all compact sets of $\sX$
  are small. Then geometric ergodicity of $P$ follows from \citet[][Lemma 15.2.8]{meyn:twee:1993}.
\end{proof}
Next, we consider MH Markov chains. The Dirac point mass at $x$ is
denoted by $\delta_x(\cdot)$. An Mtf $P$ is said to be MH type if
\begin{equation}
  \label{eq:mhtr}
  P(x, dy) = \alpha(x, y) Q(x, dy) + r(x) \delta_x(dy),
\end{equation}
where $Q$ is an Mtf with density $q(x, y)$, $\alpha$ is as given in \eqref{eq:accpro} and
\begin{equation}
  \label{eq:rx}
  r(x) = 1 - \int_{\sX} \alpha(x, y) Q(x, dy).
\end{equation}
Since $P$ in \eqref{eq:mhtr} is reversible with respect to $f$,
$f$ is its stationary distribution.  If $f(x)$ and $q(x, y)$ are
positive and continuous for all $x, y$, then from \citet[][Lemma
1.2]{meng:twee:1996} we know that the MH type Mtf \eqref{eq:mhtr}
is aperiodic, and every nonempty compact set is small. A weaker
condition is given in \cite{robe:twee:1996} that assumes $q$ is
bounded away from zero in some region around the origin. In
particular, if $f(x)$ is bounded away from $0$ and $\infty$ on
compact sets, and $\exists \; \delta_q >0, \varepsilon_q >0$ such that for
all $x, \|x-y\| \le \delta_q \Rightarrow q(x, y) \ge \varepsilon_q,$ then
$P$ given in \eqref{eq:mhtr} is $\phi$-irreducible, aperiodic and
every nonempty compact set is small.

Let $B_k(x) = \{y: \|y-x\| < k\}$ denote the open ball with center $x$
and radius $k$. Following \cite{jarn:twee:2003} an Mtf $P$ is called
random-walk-type if for any $\varepsilon >0$, $\exists \; k >0$ such that
$P(x, B_k(x)) > 1 - \varepsilon$. If $P$ is of the form \eqref{eq:mhtr}
\[
P(x, B_k(x)) = \int_{ B_k(x)} \alpha(x, y) Q(x, dy) + \int_{\sX} (1 - \alpha(x, y)) Q(x, dy) \ge \int_{ B_k(x)} Q(x, dy),
\]
then it is enough to verify $Q(x, B_k(x)) > 1 - \varepsilon$ for $P$ to be random-walk-type. We now provide some conditions for $P$ in \eqref{eq:mhtr} to be GE.
\begin{proposition}
  \label{prop:mhge}
  Suppose $P$ is of the form \eqref{eq:mhtr} and it is
  $\phi$-irreducible, aperiodic, and every nonempty compact set is
  small. If there exists a function $V: \sX \rightarrow [1, \infty]$,
  which is bounded on compact sets with
  \begin{equation}
    \label{eq:limsupp}
\limsup_{\|x\|\rightarrow\infty}\frac{PV(x)}{V(x)} < 1    
\end{equation}
and
  \begin{equation}
    \label{eq:bddrat}
\frac{PV(x)}{V(x)} \;\text{is bounded on compact sets},   
\end{equation}
then \eqref{eq:drif} holds for a small set $C$. Conversely, if $P$ is
random-walk-type and \eqref{eq:drif} holds, then $V$ satisfies
\eqref{eq:limsupp} and \eqref{eq:bddrat}.
\end{proposition}
The proof of this result can be gleaned from
\cite{jarn:hans:2000}. However, we provide a proof here for
completeness. Among other conditions, \pcite{jarn:hans:2000} Lemma 3.5 assumes
that $PV(x)/V(x)$ is bounded which is often violated as in the examples considered here.
\begin{proof}[Proof of Proposition~\ref{prop:mhge}]
  Note that, under \eqref{eq:limsupp}, \eqref{eq:drif} holds for all
  $x$ outside $C= \{x: \|x\| \le k\}$ for $k$ sufficiently
  large. Since \eqref{eq:bddrat} holds, and $V$ is bounded on compact sets, we have
  \[
    \sup_{\|x\| \le k} PV(x) \le \sup_{\|x\| \le k} \frac{PV(x)}{V(x)} \sup_{\|x\| \le k} V(x) < \infty .
  \]
  From the conditions, we know that $C$ is small. Thus,
  \eqref{eq:drif} holds. For the converse, by Lemma 2.2 of
  \cite{jarn:hans:2000} we know that every small set is bounded. Since \eqref{eq:drif} holds
  \[
    \frac{PV(x)}{V(x)} \le \lambda + \frac{b I_C(x)}{V(x)},
  \]
  implying \eqref{eq:limsupp} as $C$ is bounded and \eqref{eq:bddrat} as $b I_C(x)/V(x) \le b$.
\end{proof}
Note that
$PV(x)/V(x) = \int_{\sX} [V(y)/V(x)]\alpha(x, y) Q(x, dy) + r(x)$. As
shown in \cite{robe:twee:1996} if ess sup $r(x) =1$, then $P$ is not
GE. Necessary conditions for geometric ergodicity can also be
established by the following result of \cite{jarn:twee:2003}.
\begin{proposition}[Jarner and Tweedie]
  \label{prop:mhnec}
  If $P$ is random-walk-type with stationary density $f$, and if it is GE, then $\exists \; s>0$ such that $\E_f(\exp[s \|X\|]) < \infty$.
\end{proposition}
\section{Geometric ergodicity of the general MALA}
\label{sec:gmalage}
In this section, we study geometric convergence rates for the MH
algorithms with candidate distribution $N(c(x), h G(x))$. Thus, the proposal density is given by
\begin{equation}
  \label{eq:propgmala}
  q(x, y) = \frac{1}{(2h\pi)^{d/2} |G(x)|^{1/2}} \exp\{-(y-c(x))^\top G(x)^{-1}(y-c(x))/2h\}.  
\end{equation}
As explained in Section~\ref{sec:gmala}, distinct forms of the mean
function $c(x)$ and the covariance matrix $hG(x)$ result in the MALA and its
different variants. Let $A(x)$ denote the acceptance region, where the
proposed positions are always accepted, that is,
$A(x)=\{y:f(x)q(x,y)\leq f(y)q(y,x)\}$. If $y \in A(x)$, then
$\alpha(x, y)$ defined in \eqref{eq:accpro} is always one. Let
$R(x)=A(x)^c$ be the potential rejection region. We now define the following conditions.
\begin{itemize}
\item[A1] There exist positive definite matrices $G_1$ and $G_2$ such that $G_1 \le G(x) \le G_2 \;\forall x$.
\item[A2] The mean function $c(x)$ is bounded on bounded sets.
\item[A3] $C_1 :=\limsup_{\|x\|\rightarrow\infty}\int_{R(x)}q(x,y)(1-\alpha(x,y))dy<1$.
  \item[A4] There exists $s>0$ such that
    \begin{equation}
      \label{eq:eta}
      \eta:=\liminf_{\|x\|\rightarrow\infty}\Big(\|G_2^{-1/2}x\|-\|G_2^{-1/2}c(x)\|\Big)>\frac{\log C_2(s)-\log(1-C_1)}{s},      
    \end{equation}
where
\begin{equation}
  \label{eq:c2s}
      C_2(s)=h^{-d/2} (\pi/2)^{(d-2)/2} (|G_2|/|G_1|)^{1/2}\exp\{hs^2/2\} \int_{0}^{\infty}\exp\{-(r-hs)^2/(2h)\}r^{d-1}dr.
\end{equation}
\end{itemize}
Here, for two square matrices $G_1$ and $G_2$ having the same
dimensions, $G_1 \le G_2$ means that $G_2 - G_1$ is a positive
semi-definite matrix. That is, $G_1 \le G_2$ is the usual Loewner
order on matrices. Let $\zeta_{i+}$ and $\zeta_{i}^{+}$ be the
smallest and the largest eigenvalue of $G_i$, respectively for
$i=1,2$.
\begin{remark}
 Since $\|x\|/\sqrt{\zeta_{2}^+} \le \|G_2^{-1/2}x\| \le \|x\|/\sqrt{\zeta_{2+}}$, a sufficient condition for A4 that may be easier to check is
$ \liminf_{\|x\|\rightarrow\infty}(\|x\|/\sqrt{\zeta_{2}^+}-\|c(x)\|/\sqrt{\zeta_{2+}})>[\log C_2(s)-\log(1-C_1)]/s.$  
\end{remark}
We now state sufficient conditions for
geometric ergodicity of the MH chains with a position-dependent covariance
matrix.
\begin{theorem} 
\label{thm:gmalage}
Suppose the conditions A1--A4 hold. If $f(x)$ is bounded away from
$0$ and $\infty$ on compact sets, the MH chain with proposal density
\eqref{eq:propgmala} is GE.
\end{theorem}
\begin{remark}
  \label{rem:gmalage}
  The proof of Theorem~\ref{thm:gmalage} given in~\ref{sec:approof} uses a Lyapunov drift function
  $V_s(x)=\exp\{s\|G_2^{-1/2}x\|\}$, with $s>0$. By considering a different drift
  function $V'_s(x)=\exp\{s\|x\|\}$, $s>0$, and following the steps
  in that proof and using the fact that $G_2 \le \zeta_{2}^+ I_d$,
  another alternative for A4 can be obtained. Indeed, the condition A4
  in Theorem~\ref{thm:gmalage} can be replaced by the existence of
  $s>0$ with
  $\liminf_{\|x\|\rightarrow\infty}\big(\|x\|-\|c(x)\|\big)>[\log
  C'_2(s)-\log(1-C_1)]/s,$ where
\[
      C'_2(s)=h^{-d/2} (\pi/2)^{(d-2)/2} (\exp\{h\zeta_{2}^+s^2\}/|G_1|)^{1/2} \int_{0}^{\infty}\exp\{-(r-hs\zeta_{2}^+)^2/(2h\zeta_{2}^+)\}r^{d-1}dr.
\]
\end{remark}
\begin{remark}
  \label{rem:convinq}
  As mentioned in the Introduction, \cite{robe:twee:1996} derived conditions under
  which the MALA chain is GE. One of their conditions is `$A(\cdot)$
  converges inwards in $q$' which means
  $\lim_{\|x\|\rightarrow\infty}\int_{A(x)\Delta {\text{In}(x)}}q(x,y) dy=0$,
  where $\text{In}(x)=\{y:\|y\|\leq\|x\|\}$ and
  $A(x)\Delta \text{In}(x)=(A(x) \setminus \text{In}(x))\cup(\text{In}(x) \setminus A(x))$. Recently,
  \cite{livi:beta:byrn:2019} assume a slightly weaker condition
  $\lim_{\|x\|\rightarrow\infty}\int_{R(x)\cap \text{In}(x)}q(x,y)dy=0$ for
  establishing geometric ergodicity of Hamiltonian Monte Carlo Markov
  chains. Below we show that if A1 holds and $\|c(x)\| < M$ for all
  $x$, then $\lim_{\|x\|\rightarrow\infty}\int_{R(x)\cap \text{In}(x)}q(x,y)dy=0$ implies that $C_1 =0$,
  that is, in that case, A3 automatically holds.

  \begin{proof}[Proof of Remark~\ref{rem:convinq}]
    Since $\|c(x)\| < M$, by Cauchy-Schwartz inequality,
    \[c(x)^\top G_2^{-1}y \le \sqrt{y^\top G_2^{-1}y}\sqrt{c(x)^\top
      G_2^{-1}c(x)} \le (M/\sqrt{\zeta_{2+}})\sqrt{y^\top G_2^{-1}y}.\] Thus,
    from \eqref{eq:propgmala}, we have
    $q(x, y) \le a \exp\{-(\|G_2^{-1/2}y\|-M/\sqrt{\zeta_{2+}})^2/2h\}$ for
    some constant $a>0$. Then $C_1=0$ follows since
    \[
      C_1 \le \limsup_{\|x\|\rightarrow\infty}\int_{R(x) \cap \text{In}(x)} q(x,y)dy + \limsup_{\|x\|\rightarrow\infty}\int_{R(x) \cap \text{In}(x)^c} q(x,y)dy,
      \]
      and by DCT, the second term of the right side is zero.
  \end{proof}
\end{remark}
\begin{remark}
  \label{rem:strcon}
  For analyzing HMC algorithms, \cite{mang:smit:2021} assume that
  there exist $0 < m_2, M_2 < \infty$ such that
  $ m_2 I_d \le -\nabla^2 \log f (x) \le M_2 I_d$ for all
  $x \in \mathbb{R}^d$. A smooth target density
  $f(x) \propto \exp (- U(x))$ satisfies this condition if and only
  if $U$ is {\it $m_2$ strongly convex} and has {\it $M_2$-Lipschitz
    gradient}. Strong convexity and the existence of a Lipschitz gradient of
  $U$ are also assumed for the analysis of Langevin algorithms
  in \cite{durm:moul:2019} \cite[see also][]{dwiv:chen:wain:yu:2019}. Thus, in
  the special case of $G (x)= (-\nabla^2 \log f (x))^{-1}$, which is
  often used in practice for implementing the MMALA
  \citep{giro:cald:2011}, A1 is same as the assumption of
  \cite{mang:smit:2021} mentioned above.
\end{remark}

\begin{remark}
  \label{rem:choiceG}
  As discussed in \cite{giro:cald:2011}, for implementing the MMALA
  and the PMALA in Section~\ref{sec:sglmm}, we use $G=\mathscr{I}^{-1}$,
  the expected Fisher information matrix plus the negative Hessian of
  the logarithm of the prior density. For such a $G$, we show that A1
  holds for the popular binomial-logit link GLMM, and
  Theorem~\ref{thm:gmalage} is used to establish a CLT for these Markov chains. On the other
  hand, for establishing consistency of $\bar{g}_n$ for the adaptive
  Metropolis algorithm, \cite{haar:saks:tamm:2001} assume that the
  proposal covariance matrix $G_n$ satisfies
  A1 even for the
  target density that is bounded from above and has bounded support.
\end{remark}

\begin{remark}
  If $c(x)$ is a continuous function of $x$, then A2 holds. For example,
  for the MALA or the PCMALA if $\nabla\log f(x)$ is continuous, then A2
  holds.
\end{remark}

\begin{remark}
  From \eqref{eq:c2s},
  $C_2(s) = h^{-d/2} (\pi/2)^{(d-2)/2} (|G_2|/|G_1|)^{1/2}
  \int_{0}^{\infty}\exp\{-r^2/(2h) + rs\}r^{d-1}dr$. Thus, $C_2(s)$ is
  increasing in $s$. \cite{robe:twee:1996} considered the MALA chain.
  When $d=1$, for the MALA chains, $G_1 = 1 = G_2$ and
  $C_2(s) = (h\pi/2)^{-1/2} \int_{0}^{\infty}\exp\{-r^2/(2h) +
  rs\}dr$. So $\lim_{s \rightarrow 0} C_2(s) =1$. From
  Remark~\ref{rem:convinq}, we know that under \pcite{robe:twee:1996}
  `$A(\cdot)$ converges inwards in $q$' condition, we have $C_1=0$, so the
  condition \eqref{eq:eta} is equivalent to $\exp(s \eta) >
  C_2(s)$. On the other hand, \pcite{robe:twee:1996}
  other condition for the MALA chain to be GE is $\eta > 0$.
\end{remark}

In the proof of Theorem~\ref{thm:gmalage} we have worked with the
drift function $V_s(x)=\exp\{s\|G_2^{-1/2} x\|\}$, with $s>0$. Using a
different drift function we establish the following theorem providing
a slightly different condition for geometric ergodicity. Let us define another condition:
\begin{itemize}
  \item[A5]
      $\limsup_{\|x\|\rightarrow\infty}(\|c(x)\|^2/\|x\|^2)<(1-C_1)(|G_1|/|G_2|)^{1/2}.$      
\end{itemize}

\begin{theorem}
  \label{thm:gmalage2}
  Suppose the conditions A1--A3 and A5 hold. If $f(x)$ is bounded away from
$0$ and $\infty$ on compact sets, the MH chain with proposal density
\eqref{eq:propgmala} is GE. 
\end{theorem}

\begin{remark}
  If the growth rate of $\|c(x)\|$ is smaller than that of $\|x\|$, 
  then A4 and A5 hold (see e.g. the binomial SGLMM example in
  Section~\ref{sec:sglmm}). In this case, $C_1$ does not need to be
  explicitly found to be used within A4 or A5. On the other hand, if $\eta$
  can be derived, then a grid search for $s$ can be done to verify A4.
\end{remark}

\begin{remark}
  \label{rem:choiceprop}
  Although a Gaussian proposal density \eqref{eq:propgmala} is assumed
  in Theorems~\ref{thm:gmalage} and \ref{thm:gmalage2}, following the
  proofs of these results, one may try to establish conditions for
  geometric ergodicity for other proposal densities as long as upper
  bounds to the means of the drift functions with respect to these
  densities can be
  derived.
\end{remark}
We now provide some general conditions under which an MH algorithm
with proposal density \eqref{eq:propgmala} does not produce a GE
Markov chain. Recall that for the MALA chain and its variants, the
mean function $c(x)$ is of the form $x + h e(x)$ for some function
$e(x)$ and step-size $h$. For the rest of this section, we assume
$c(x) = x + h e(x)$.

\begin{theorem}
  \label{thm:necmgf}
  If A1 holds and $\|e(x)\| < M$ for all
  $x$ and for some $M>0$, then a necessary condition for geometric ergodicity of the MH chain
with proposal density \eqref{eq:propgmala} is $\E_f(\exp[s \|X\|]) < \infty$ for some $s >0$.  
\end{theorem}
The following theorem provides another necessary condition for geometric ergodicity of the MH chain
with proposal density \eqref{eq:propgmala}.
\begin{theorem}
\label{thm:gmalanclt}
If $f(\cdot)$ is bounded, $A1$ holds and 
\begin{equation}
  \label{eq:gmalanclt}
    \liminf_{\|x\|\rightarrow\infty}\frac{\|e(x)\|}{\|x\|}>\frac{2}{h},
\end{equation}
then the MH chain with proposal density \eqref{eq:propgmala} is not GE.
\end{theorem}

\section{Geometric ergodicity of the PCULA}
\label{sec:gegula}
Based on the Langevin diffusion \eqref{eq:lang}, \cite{robe:twee:1996}
considered the discrete time Markov chain $\{X_n\}_{n \ge 0}$ given by
\begin{equation}
  \label{eq:ula}
    X_n|X_{n-1}\sim N(X_{n-1}+ (h/2)\nabla\log f(X_{n-1}),hI_d),
\end{equation}
where $f(\cdot)$ is the target density. \eqref{eq:ula} is
referred to as the unadjusted Langevin algorithm (ULA). For molecular
dynamics applications, the algorithm was considered before \cite[see
e.g.][]{erma:1975}. However, as mentioned before, when the coordinates
are highly correlated, the same step-size for all directions may not be
efficient. Therefore, we consider the pre-conditioned unadjusted
Langevin algorithm (PCULA) by replacing the identity matrix in
\eqref{eq:ula} with $G$ that takes the correlation of different
coordinates into consideration:
\begin{equation}
  \label{eq:pcula}
    X_n|X_{n-1}\sim N(X_{n-1}+(h/2)G\nabla\log f(X_{n-1}),hG).
  \end{equation}
  Geometric convergence of the ULA chain \eqref{eq:ula} in the special
  case when $d=1$ is considered in \cite{robe:twee:1996}. Recently,
  \cite{durm:moul:2019} provide some non-asymptotic results for the ULA
  with non-constant step-sizes in the higher
  dimensions \cite[see also][]{durm:moul:2017,vemp:wibi:2019}. \cite{durm:moul:2019} also compare the performance of
  the PCULA chains with the PCMALA chains in the context of a Bayesian
  logistic model for binary data.  From Section~\ref{sec:gmalage} we can derive conditions for
  geometric ergodicity of the Markov chain driven by
  $X_n | X_{n-1}=x \sim N(c(x), h G)$. Note that, in the absence of an
  accept-reject step, $C_1=0$ for the PCULA chain. Although PCULA
  avoids the accept-reject step, it is important to note that its
  equilibrium distribution is no longer $f$.
  \begin{proposition}
    \label{prop:pcula}
    Let $c(x)$ be a continuous function of $x$. If A4 or A5 holds with
    $C_1 =0$ and $G_1=G=G_2$, then the Markov chain $\{X_n\}_{n \ge 0}$ given by $X_n | X_{n-1}=x \sim N(c(x), h G)$ is GE.
  \end{proposition}  
  If $f$ is Gaussian with $f(x)\propto \exp\{-x^{\top}W^{-1}x/2\}$,
  then $\nabla\log f(x)=-W^{-1}x$. In this case, the PCULA Markov
  chain \eqref{eq:pcula} is given by:
\[
   X_n =X_{n-1}-(h/2)GW^{-1}X_{n-1}+\sqrt{h}G^{1/2}\epsilon_n =AX_{n-1}+\sqrt{h}G^{1/2}\epsilon_n,
\]
where $A=I-(h/2)GW^{-1}$ and $\epsilon_n \stackrel{iid}{\sim} N(0,I_d)$.  We can further
extend it by considering more general forms of $\nabla\log f(x)$. In particular, we consider the Markov Chain:
\begin{align}
\label{eq:pula_extend}
    X_n =AX_{n-1}+e(X_{n-1})+\sqrt{h}G^{1/2}\epsilon_n
\end{align}
where $e(x)$ is a continuous function.
\begin{corollary}
  \label{cor:pcula}
  If
  \begin{equation}
    \label{eq:pculaeig}
    \limsup_{\|x\|\rightarrow\infty}\frac{\|e(x)\|^2+ 2x^{\top}e(x)}{\|x\|^2}<1-\lambda^+,
  \end{equation}
  where $\lambda^+$ is the largest eigenvalue of $A^{\top}A$, then the Markov chain given in \eqref{eq:pula_extend} is GE.
\end{corollary}
\begin{proof}[Proof of Corollary~\ref{cor:pcula}]
  Since
  \[
    \frac{\|Ax + e(x)\|^2}{\|x\|^2} \le \lambda^+ + \frac{e(x)^{\top}e(x)+ 2x^{\top}e(x)}{\|x\|^2},
  \]
  the proof follows from Proposition~\ref{prop:pcula} as A5 holds.
\end{proof}
\begin{remark}
  \label{rem:pculasing}
   If $\lambda^+ <1,$ and $\|e(x)\|=o(\|x\|)$, then \eqref{eq:pculaeig} holds.
 \end{remark}
 Note that, $\lambda^+ < 1$ is equivalent to that the singular values
 of $A$ are strictly less that one. On the other hand, if
 $GW^{-1} = W^{-1}G$ then if $\rho \in (0, 4/h)$, where $\rho$ is any
 eigenvalue of $GW^{-1}$, then $\lambda^+ < 1$.
\begin{remark}
For the ULA chain \eqref{eq:ula}, $A=I$ and $e(x)=(h/2)\nabla\log f(x)$. Thus when $d=1$, \eqref{eq:pculaeig} becomes  
\begin{align*}
    \limsup_{|x|\rightarrow\infty}\frac{(h\nabla\log f(x)/2)^2+hx\nabla\log f(x)}{x^2}<0.
\end{align*}
On the other hand, a sufficient condition given in \citet[][Theorem 3.1 ]{robe:twee:1996} is that
$    \lim_{|x|\rightarrow\infty}h\nabla\log f(x)/[2x]<0\;\;
\text{and} \;\;
(1+\lim_{x\rightarrow \infty} h\nabla\log f(x)/[2x])(1-\lim_{x\rightarrow -\infty} h\nabla\log f(x)/[2|x|])<1.$
\end{remark}

\section{Generalized linear mixed models}
\label{sec:sglmm}
GLMMs are popular for analyzing different types of correlated
observations. Using unobserved Gaussian random effects, GLMMs permit
additional sources of variability in the data. Conditional on the
random effect $x= (x^{(1)}, \dots, x^{(m)})$, the response/observation
variables $\{Z_1, \dots, Z_m\}$ are assumed to be independent with
$Z_i|x^{(i)} \stackrel{ind}{\sim} a(z_i; \mu_i)$, where
the conditional mean $\mu_i=\E(Z_i|x^{(i)})$ is related to $x^{(i)}$
through some link functions. Since $Z_1,\dots,Z_m$ are conditionally
independent, the joint density of $z=(z_1, \dots, z_m)$ is
$a(z;\mu)=\prod_{i=1}^ma(z_i;\mu_i).$ Here, we consider the two most
popular GLMMs, namely the binomial GLMM with the logit link and the
Poisson GLMM with the log link. For the binomial-logit link model,
$a(z_i;\mu_i)={\ell_i \choose
  z_i}(\mu_i/\ell_i)^{z_i}(1-\mu_i/\ell_i)^{\ell_i-z_i},
z_i=0,1,\dots,\ell_i,$ with $\log(\mu_i/[\ell_i-\mu_i])=x^{(i)}$. Whereas,
for the Poisson-log link model,
$a(z_i;\mu_i)= \exp(-\mu_i)\mu_i^{z_i}/z_i!, z_i=0,1,\dots,$ with
$\log(\mu_i) =x^{(i)}$.

The likelihood functions of GLMMs are not available in closed form,
but only as a high dimensional integral, that is,
$ L(z) = \int_{\mathcal{R}^m} a(z;\mu) p(x) dx$ where $p(x)$ is the
multivariate Gaussian density for $x$ with mean $D\beta$ and
covariance matrix $\Sigma$. Here $\beta$ and $D$ are the fixed effects
and the fixed effects design matrix, respectively.  In this section,
we assume that $(\beta, \Sigma)$ are known, and consider exploring the
target density
\begin{equation}
  \label{eq:post}
  f(x) \equiv f(x|z) = \Bigg[\prod_{i=1}^m a(z_i;\mu_i) p(x)\Bigg]\Bigg/L(z),
\end{equation}
using the different variants of the MALA and the ULA discussed in Sections~\ref{sec:gmala} and \ref{sec:gegula}.


As mentioned in Remark~\ref{rem:choiceG}, for the MMALA we
use $G(x)= \mathscr{I}^{-1}(x)$ where $\mathscr{I}=-\nabla^2\log f$.
Thus, we begin with differentiating $\log f$ for the binomial-logit link model. Note that, in this case, $\log f(x)$ (up to a constant) is 
\begin{equation*}
    -\frac{m\log(2\pi)+\log|\Sigma|}{2} +\sum_{i=1}^m \Big[\log{\ell_i \choose
  z_i} + z_i x^{(i)}-\ell_i\log(1+\exp(x^{(i)}))\Big]-\frac{(x-D\beta)^{\top}\Sigma^{-1}(x-D\beta)}{2}.
\end{equation*}
Letting $\ell=(\ell_1,\dots,\ell_m)$, we have
\begin{equation}
  \label{eq:binom12der}
    \frac{\partial\log(f(x))}{\partial x}=z-\ell \cdot\frac{e^x}{1+e^x}-\Sigma^{-1}(x-D\beta),\;   \frac{\partial^2\log(f(x))}{\partial x^2}=\text{diag}\bigg(-\ell\cdot \bigg\{\frac{e^x}{1+e^x}-\bigg[\frac{e^x}{1+e^x}\bigg]^2\bigg\}\bigg)-\Sigma^{-1},
\end{equation}
and
\begin{equation}
  \label{eq:binom3der}
    \frac{\partial^3\log(f(x))}{\partial x^3}=\text{diag}\bigg(-\ell\cdot\bigg\{\frac{e^x}{1+e^x}-3\bigg[\frac{e^x}{1+e^x}\bigg]^2+2\bigg[\frac{e^x}{1+e^x}\bigg]^3\bigg\}\bigg).
\end{equation}
In the above diag $(z)$ denotes the $m \times m$ diagonal matrix with
diagonal elements $z$. Since $\nabla^3 \log f(x)$ in \eqref{eq:binom3der} is a diagonal
matrix, from the proposition in \cite{xifa:sher:2014}, it follows that
the PMALA with $G(x)= \mathscr{I}^{-1}(x)$ is the same as
the MMALA in this case. Indeed, in this
case, in \eqref{eq:plang},
\begin{equation}
  \label{eq:gambin}
\Gamma_i(x) = \sum_j \partial \mathscr{I}^{-1}_{ij} (x)/\partial
x^{(j)} = - \sum_j \mathscr{I}^{-1}_{ij} (x) (\partial \mathscr{I}_{jj} (x)/\partial
x^{(j)}) \mathscr{I}^{-1}_{jj} (x).  
\end{equation}
 For the PCMALA, the covariance matrix $G$
does not depend on the current position $x$. In
Section~\ref{sec:numer} we consider several choices of $G$.

\begin{theorem}
  \label{thm:gmalabinom}
  For the binomial
  GLMM with the logit link, for appropriate values (given in the proof of this result) of $h$,
the PCMALA, the MMALA and the PCULA Markov chains are GE.
\end{theorem}

\begin{remark}
  When $G=I$, \cite{chri:moll:waag:2001} established that if
  $h\in (0,2)$ then a `truncated' MALA chain for the binomial-logit link model is GE. On the other hand, using
  Theorem~\ref{thm:gmalabinom}, geometric ergodicity of this
  chain can be shown to hold when $h\in(0,4)$.
\end{remark}
Next, we derive $\log f$ for the Poisson
GLMMs with the log link. In this case, 
\begin{equation}
   \label{eq:poislogden}
\log f(x)= \mbox{a constant} + \sum_{i=1}^m (z_ix^{(i)}-\exp\{x^{(i)}\}) -(x-D\beta)^{\top}\Sigma^{-1}(x-D\beta)/2,
\end{equation}
\begin{equation}
  \label{eq:poisder}
    \nabla\log f(x)=z-\exp\{x\}-\Sigma^{-1}(x-D\beta),\; \nabla^2\log f(x)=-\text{diag}(\exp\{x\})-\Sigma^{-1},
\end{equation}
and $\nabla^3\log f(x)=-\text{diag}(\exp\{x\})$.
\begin{proposition}
  \label{thm:pcmalapois}
  For the Poisson GLMM with the log link, the PCMALA chain is not GE for any $h \in (0,\infty)$ and any pre-conditioning matrix $G$. 
\end{proposition}
A proof of Proposition~\ref{thm:pcmalapois} for the MALA chain can be found in \cite{chri:moll:waag:2001}.

\subsection{Numerical examples of SGLMMs}
\label{sec:numer}

Spatial generalized linear mixed models (SGLMMs),
introduced by \cite{digg:tawn:moye:1998}, are often used for analyzing
non-Gaussian spatial data that are observed in a continuous region
\cite[see e.g.][]{zhan:2002, roy:evan:zhu:2016, evan:roy:2019}. SGLMMs
are GLMMs where the random effects consist
of a spatial process. Conditional on the spatial process, the response
variables are assumed to follow a distribution which only depends on
the site-specific conditional means. As in the GLMMs, a link function relates the means
of the response variable to the underlying spatial process. 

Let $\{X(s), s \in \Sbb\}$ be a Gaussian random field with mean
function $\E(X(s))$ and the covariance function
$\Cov(X(s), X(s')) = \sigma^2 \rho_\theta(\|s-s'\|)$. Here, the
parameter $\sigma^2$ is called the partial sill, and some examples of
the parametric correlation functions $\rho_\theta$ are the
\textit{exponential}, the \textit{Mat\'{e}rn}, and the
\textit{spherical} families. The mean $\E(X(s))$ is generally a
function of some regression parameters $\beta$, and the known location
dependent covariates. Conditional on the realized value of the
Gaussian random field, $\{x(s), s \in \Sbb\}$, and for any
$s_1, \dots, s_m \in \Sbb$, the response variables
$\{Z(s_1), \dots, Z(s_m)\}$ are assumed to be independent with
$Z(s_i)|x(s_i) \stackrel{ind}{\sim} a(z_i; \mu_i)$, where
$\mu_i=\E(Z_i|x(s_i))$ is the conditional mean related to $x(s_i)$
through some link function. Denoting $z(s_i)$ ($x(s_i)$) simply by
$z_i (x^{(i)})$ we arrive at the SGLMM target density $f(x)$ as given
in \eqref{eq:post}.

In this section, we perform simulation studies to assess the
performance of the PCMALA, the PMALA (MMALA) and the PCULA with
different choices of the pre-conditioning matrix $G$ in the context of
conditional simulation in SGLMMs. We also compare the performance of
these algorithms with the random walk Metropolis (RWM) algorithms. The
domain for the simulations is fixed to $\Sbb=[0,1]^2$, and the
Gaussian random field $x$ is considered at an $21 \times 21$ square
grid covering $\Sbb$. A realization of the data $z$ consists of
observations from the binomial spatial model at $m=350$ randomly
chosen sites with number of trials $\ell_i=50$ for all
$i=1,\dots,350$. The mean of the random field is set to $1.7$ for the
left half of the domain and to $-1.7$ for the right half, while its
covariance is chosen from the exponential family
$\text{Cov}(x(s),x(s'))= \sigma^2 \exp\{-\| l-l'\|/\phi \}$, with
$\sigma^2 =1$ and range $\phi=0.5$. We also consider simulated data
from the Poisson-log SGLMM and other setup as above.

For the PCMALA and the PCULA we consider four choices of $G$: i)
$G=I$, which corresponds to simply the MALA and the ULA, respectively,
ii) $G=\Sigma$, the covariance matrix of the (prior) distribution of
$x$, iii) $G= \text{diag} (\hat{\mathscr{I}}^{-1})$, the diagonal
matrix with diagonal elements from $\hat{\mathscr{I}}^{-1}$ where
$\hat{\mathscr{I}}=-\nabla^2\log f|_{x=\hat{x}},$ with
$\hat{x}=\text{argmax}_{x\in R^m}\log f(x)$, and finally iv)
$G= \hat{\mathscr{I}}^{-1}$. Note that for the RWM, the candidate
proposed position is $y=x+\sqrt{h}G^{1/2}\epsilon$ with
$\epsilon \sim N(0, I_m)$. For the RWM algorithms, we also consider
the before mentioned four choices of the $G$ matrix. The step-size $h$
is selected using pilot runs of the chains, ensuring the acceptance
rate for the different algorithms falls in (60\%, 70\%). Thus,
together we consider nine MH algorithms--four RWM chains denoted as
RWM1, RWM2, RWM3 and RWM4 corresponding to the four choices of the $G$
matrix in the before mentioned order, four PCMALA chains PCMALA1,
PCMALA2, PCMALA3 and PCMALA4 with the above $G$ matrices, respectively
and the PMALA chain.  Similarly, we consider four PCULA chains denoted
by PCULA1, PCULA2, PCULA3 and PCULA4 corresponding to the four $G$
matrices in the before mentioned order.

The empirical performance of the different MCMC algorithms is compared
using several measures (See \cite{roy:2020} for a simple introduction
to some of these convergence diagnostic measures.). In particular, the
MCMC samplers are compared using lag $k$ autocorrelation function
(ACF) values, the effective sample size (ESS) and the multivariate ESS (mESS),
ESS (mESS) per unit time, the mean squared jump distance (MSJD),
and the multivariate potential scale reduction factor (MPSRF). 
As mentioned in \cite{roy:2020}, for fast-mixing Markov chains, lag
$k$ ACF values drop down to (practically) zero quickly as $k$
increases, whereas high lag $k$ ACF values for larger $k$ indicate
slow mixing of the Markov chain. In one dimensional setting, ESS is
defined as $ \text{ESS}=n\hat{\lambda}^2_g/\hat{\sigma}^2_g,$ where
$n$ is the length of the chain, $\hat{\sigma}^2_g$ is the estimated
variance in the CLT as mentioned in Section~\ref{sec:mcbg} and
$\hat{\lambda}^2_g$ is the sample variance. When $g$ is a
$\mathbb{R}^p$ valued function for some $p>1$,
\cite{vats:fleg:jone:2019} define mESS as
$\text{mESS}=n(|\widehat{\Lambda}_g|/|\widehat{\Sigma}_g|)^{1/p},$
where $\widehat{\Lambda}_g$ is the sample covariance matrix and
$\widehat{\Sigma}_g$ is the estimated covariance matrix from the
CLT. From the definition of the ESS and the mESS, we see that larger
values of these measures imply higher efficiency of the Markov
chain. The ESS and mESS are calculated using the R package {\it
  mcmcse}. The MSJD based on $n$ iterations of a Markov chain
$\{X_n\}$ is defined as
$\text{MSJD}:= \sum_{i=1}^{n-1}\|X_{i+1}-X_{i}\|^2/(n-1)$.  MSJD
compares how much the chains move around the space, and larger values
indicate higher amount of mixing. As mentioned in
\cite{broo:gelm:1998}, starting at overdispersed initial points, if
the MPSRF $\hat{R}_p$ is sufficiently close to one, then the
simulation can be stopped. Thus, Markov chains for which $\hat{R}_p$
reaches close to one faster are preferred. We use the R package {\it
  coda} for computing $\hat{R}_p$. As mentioned in \cite{roy:2020},
for using most of the above mentioned numerical measures including
ESS, mESS, and MPSRF, existence of a Markov chain CLT is assumed
emphasizing the importance of establishing the geometric ergodicity
properties of the paper. While for the binomial-logit link SGLMM we have
established a CLT for the PCMALA and MMALA chains, for the Poisson model we are
naively going to use the before mentioned numerical measures to compare
the different algorithms.

\subsection{Comparison of the adjusted Langevin algorithms}
\label{subsec:comp_pmala}
We ran each of the nine MH chains started at $x_{\text{true}}$, the
`true' value of $x$ used to simulate the data $z$, for 150,000
iterations. For the binomial SGLMM, Table~\ref{table:ess.bl} provides
the ESS values for the three marginal chains corresponding to
$(x^{(1)}, x^{(175)}, x^{(350)})$, the first, the 175th and the 350th
element of the $350$ dimensional $x$ vector at the randomly chosen
sites mentioned before. The locations for these three points in the
$21 \times 21$ square grid covering $\Sbb$ are (0,0), (0.1, 0.5), and
(1, 1), respectively. Table~\ref{table:ess.bl} also includes the mESS
values for the multivariate $350$ dimensional Markov chains. From
Table~\ref{table:ess.bl}, we see that the choice of the covariance
matrix $G$ does not change the performance of the RWM algorithms much,
whereas efficiency of the PCMALA can vary greatly with $G$. Indeed,
when $G= \hat{\mathscr{I}}^{-1}$, there are huge gains in efficiency
for the PCMALA resulting in much higher ESS, mESS values compared to
the other choices of $G$. We see that even with the ideal choice of
$G$, that is, $G= \mathscr{I}^{-1}$, the PMALA has much smaller ESS
and mESS values than the PCMALA with $G= \hat{\mathscr{I}}^{-1}$ but
it is better than the PCMALA with non-optimal choices of $G$, such as
$G=I$, $G=\Sigma$ or $G=\text{diag}(\hat{\mathscr{I}}^{-1})$. On the
other hand, for the PMALA, unlike the PCMALA, the covariance matrix
$G$ needs to be recomputed in every iteration, leading to higher
computational burden. This is why, the improvement of the PMALA over
the PCMALA with $G=I$, $G=\Sigma$ or
$G=\text{diag}(\hat{\mathscr{I}}^{-1})$ in terms of time-normalized
efficiency (ESS per minute) reduces. The PCMALA with $G=\hat{\mathscr{I}}^{-1}$ results in much higher
values of ESS, mESS and ESS/min than the
other algorithms considered here. Indeed, the PCMALA with
$G=\hat{\mathscr{I}}^{-1}$ results in {\it more than 20 times}
equivalent independent samples than the PMALA for the same amount of
running time.

Table~\ref{table:msejd.bl} provides the MSJD values for the nine
chains. Again, for the RWM, the MSJD values remain similar regardless
of the choice of $G$. The PMALA has higher MSJD values than the PCMALA
with $G=I$, $G=\Sigma$ or $G=\text{diag}(\hat{\mathscr{I}}^{-1})$
implying better mixing, whereas with $G=\hat{\mathscr{I}}^{-1}$ PCMALA
dominates the PMALA and the RWM
algorithms. Figure~\ref{fig:acf.binomlogit} shows the ACF plots for
the first 50 lags for the nine MH algorithms for each of the three
marginal chains. The ACF plots corroborate faster mixing for the
PCMALA chains with $G=\hat{\mathscr{I}}^{-1}$ than all
other eight chains and MMALA than the other Markov chains except PCMALA with $G=\hat{\mathscr{I}}^{-1}$. Indeed,
only for the PCMALA chain with $G=\hat{\mathscr{I}}^{-1}$, the lag $k$
autocorrelation becomes negligible by $k=50$.

Next, for each of the nine MH algorithms, we compute the MPSRF
$\hat{R}_p$ from five parallel chains started from $x_{\text{true}}$,
$- x_{\text{true}}$, $0$ (a vector of zeros) and
$x_{\text{true}} \pm 1$, respectively. The $\hat{R}_p$ plots are given
in Figure~\ref{fig:gelmandiag.binomlogit}. From these plots we see
that for the PCMALA chain with $G=\hat{\mathscr{I}}^{-1}$, $\hat{R}_p$
reaches below 1.1 (a cutoff widely used by MCMC practitioners)
before 5,000 iterations, whereas for several other algorithms,
including the PMALA, $\hat{R}_p$ is still larger than 1.1 even after
80,000 iterations. Thus, as for the other diagnostics, $\hat{R}_p$
also indicates superior performance of the PCMALA chain with
$G=\hat{\mathscr{I}}^{-1}$ than the other MH algorithms considered
here and MMALA is the second best. 
The performance of the nine MCMC algorithms for the Poisson-log
link SGLMM, as observed from the tables and figures given in~\ref{sec:apppois}, is
similar to the binomial-logit link SGLMM discussed here.

We considered other values of $m$ as well. For smaller $m$ (less than
50), we observe that the same or similar step-size $h$ can be used
for both PCMALA with $G=\hat{\mathscr{I}}^{-1}$ and MMALA to achieve similar acceptance rates and in
these lower dimensions, PCMALA with $G=\hat{\mathscr{I}}^{-1}$ has slightly better or similar
performance as the MMALA. On the other hand, in the higher dimensions
as we present here, MMALA needs much smaller $h$ to attain a similar
acceptance rate as the PCMALA with $G=\hat{\mathscr{I}}^{-1}$. The small step-size, in turn, leads to
more correlated samples and smaller ESS values for the MMALA in the
higher dimensions.

\begin{table}[H]
\caption{ESS values for the MH chains for the binomial SGLMM with the logit link}
\centering
\begin{tabular}{c c c c  c} 
 \hline\hline
Algorithm & $G$ matrix &ESS(1, 175, 350) & ESS/min& mESS \\  [0.5ex]
 \hline\hline

  \multirow{4}{4em}{RWM}& $I$ & ( 44,35,48 ) &( 0.20,0.16,0.22 )  & 1,064 \\ 
& $\Sigma$ & ( 40,21,14 ) &( 0.17,0.09,0.06 ) & 1,074 \\
& diag $\hat{\mathcal{I}}^{-1}$ & ( 28,30,28 ) & ( 0.12,0.13,0.12 ) &  1,070 \\
& $\hat{\mathcal{I}}^{-1}$ & ( 37,42,68 ) & ( 0.15,0.18,0.28 ) &  1,055\\
\hline
\multirow{4}{4em}{PCMALA}& I & ( 8,10,6 ) & ( 0.04,0.05,0.03 ) &  1,051\\
& $\Sigma$&( 9,7,6 ) &( 0.05,0.04,0.03 ) &  1,039\\
& diag $\hat{\mathcal{I}}^{-1}$ &( 204,245,198 )&( 1.08,1.29,1.05 )& 1,274\\
& $\hat{\mathcal{I}}^{-1}$ &( 9,249,8,282,9,066 ) &( 48.95,43.83,47.98 ) &  12,422\\
PMALA&  & ( 664,792,834 ) &( 2.12,2.52,2.66 )  &  2,623\\
\hline\hline
\end{tabular}
\label{table:ess.bl}
\end{table}

\begin{table}[H]
\caption{MSJD values for the MH chains for the binomial SGLMM with the logit link}
\centering
\begin{tabular}{c c c c c c c c c c} 
 \hline\hline
RWM1 &RWM2 &RWM3 & RWM4&PCMALA1&PCMALA2&PCMALA3&PCMALA4&PMALA\\  
 \hline\hline
  0.024 &0.027 &0.017 &0.018&2.19e-06&2.28e-09 &0.15&5.16& 0.496\\
\hline\hline
\end{tabular}
\label{table:msejd.bl}
\end{table}

\begin{figure*}[h]
  \includegraphics[width=\linewidth]{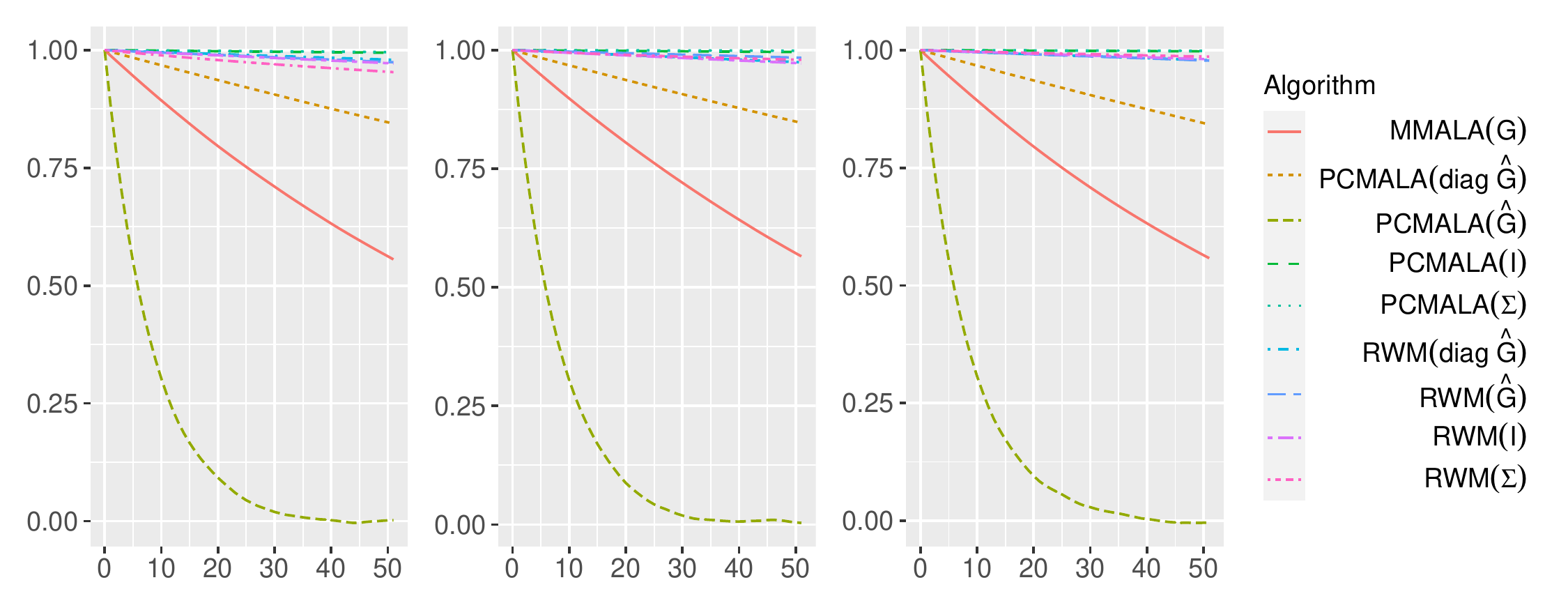}
  \caption{ACF plots for $x^{(1)}$ (left panel), $x^{(175)}$ (center
    panel), and $x^{(350)}$ (right panel) for the MH chains for the binomial SGLMM with the logit
    link. In the legend, $G$ refers to $\mathscr{I}^{-1}$ and $\hat{G}$ refers to $\hat{\mathscr{I}}^{-1}$.}
\label{fig:acf.binomlogit}
\end{figure*}

\begin{figure*}[h]
   \begin{minipage}[b]{0.24\linewidth}
    \includegraphics[width=\linewidth]{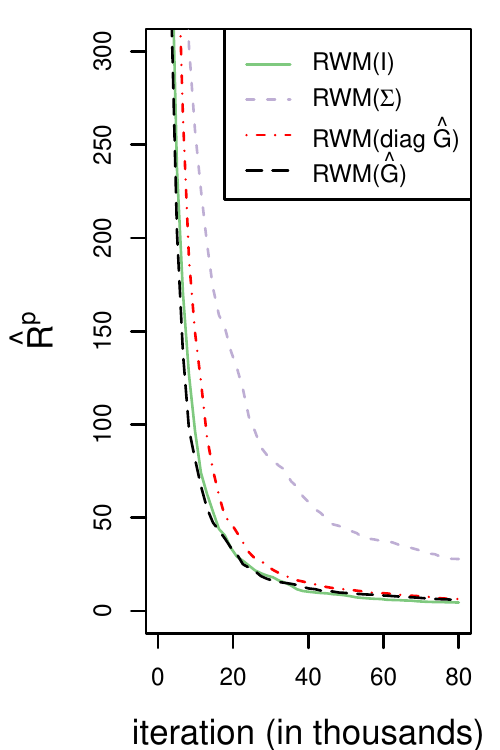}
  \end{minipage}
  \begin{minipage}[b]{0.24\linewidth}
    \includegraphics[width=\linewidth]{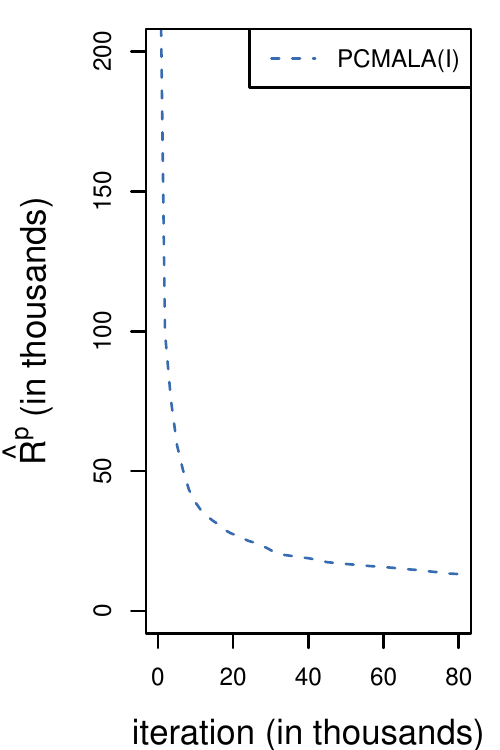}
  \end{minipage}
  \begin{minipage}[b]{0.24\linewidth}
    \includegraphics[width=\linewidth]{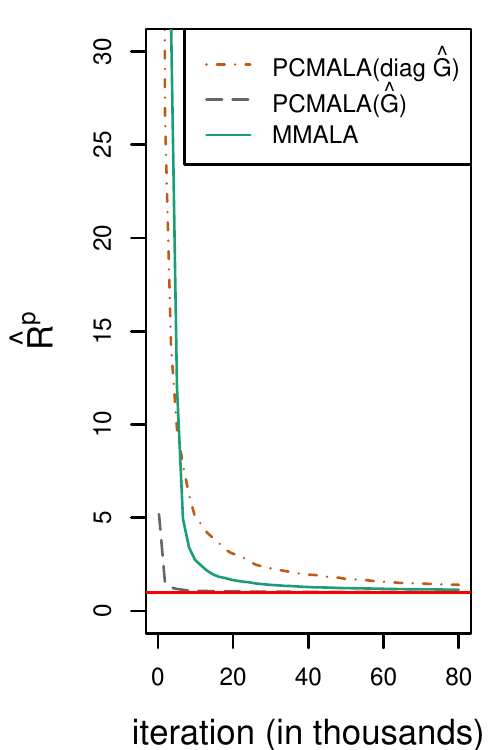}
  \end{minipage}
  \begin{minipage}[t]{0.24\linewidth}
    \includegraphics[width=\linewidth]{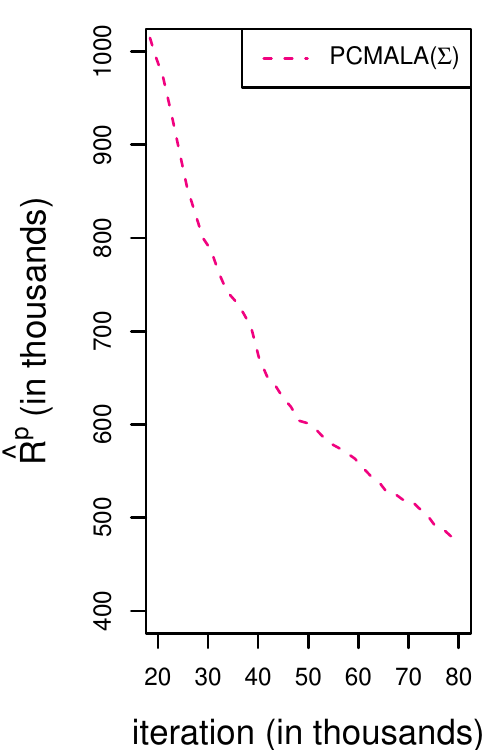}
  \end{minipage}
 \caption{Gelman and Rubin's $\hat{R}_p$ plot from the five parallel MH chains for the binomial SGLMM with the logit
    link. The red horizontal line on the third plot from the left has unit height. In the legend, $G$ refers to $\mathscr{I}^{-1}$ and $\hat{G}$ refers to $\hat{\mathscr{I}}^{-1}$.}
\label{fig:gelmandiag.binomlogit}
\end{figure*}

\subsection{Comparison of the pre-conditioned unadjusted Langevin algorithms}
\label{subsec:comp_pcula}
In this section, we compare the four PCULA chains mentioned before in
the context of simulated data from the binomial and Poisson
SGLMMs. Since the unique stationary density of each of these PCULA is
different, we do not use ESS for comparing these chains. As in
Section~\ref{subsec:comp_pmala}, we ran each of the PCULA chains for
150,000 iterations starting at
$x_{\text{true}}$. Table~\ref{table:msejd.PCULA.bl} provides the MSJD
values for the PCULA chains for the binomial and the Poisson
SGLMMs. As for the PCMALA, we see that when the pre-conditioning
matrix $G$ is $\hat{\mathscr{I}}^{-1}$, the PCULA chain results in
higher mixing than the other PCULA
chains. Figures~\ref{fig:acf.PCULA.binomlogit} and
\ref{fig:acf.PCULA.poissonlog} provide the ACF values for the first 50
lags. For the binomial model, we see that except when $G=\Sigma$, for
the other PCULA chains, the ACF values drop down quickly. Also, for
the binomial SGLMM, for smaller lags, PCULA4 has slightly higher ACF
values than PCULA1 ($G=I$). Recall that, if $G=I$, the PCULA boils
down to the ULA. For the Poisson SGLMM, for PCULA4, the ACF values
(practically) drop down to zero before five lags, whereas, the ACF
values for the other three PCULA are quite large even after 50
lags. Thus, as for the adjusted Langevin algorithms, the
pre-conditioning matrix $\hat{\mathscr{I}}^{-1}$ results in better
PCULA than the other choices of $G$ considered here.
 Finally, Figure~\ref{fig:gelmandiag.pcula} 
 provides the $\hat{R}_p$ plots based on
the five Markov chains started at the same five points $x_{\text{true}}$,
$- x_{\text{true}}$, $0$ and $x_{\text{true}} \pm 1$ as in
Section~\ref{subsec:comp_pmala}, for each of the four PCULA
chains. For both binomial and Poisson SGLMMs, the $\hat{R}_p$
reaches below 1.1 before 5,000 iterations of the PCULA4 chain. The
PCULA3 algorithm ($G= \text{diag} (\hat{\mathscr{I}}^{-1})$) is the
second best performer in terms of $\hat{R}_p$.

\section{Discussions}
\label{sec:disc}
In this paper, we establish conditions for geometric convergence of
general MH algorithms with normal proposal density involving a
position-dependent covariance matrix. If the mean of the proposal
distribution is of the form $x + h e(x)$, where $x$ denotes the
current state, the users implementing these MCMC algorithms should
make sure that $\|e(x)\|$ does not grow too fast with
$\|x\|$. Similarly, if $\|e(x)\|$ shrinks, then the tails of $f(x)$
need to die down rapidly. As special cases, our results apply to the
MMALA and other modern variants of the MALA. For the MMALA and other
MALA chains, first and higher-order derivatives of the log target
density are required. Here, in our GLMM examples, the derivatives are
available in closed form. \cite{giro:cald:2011} discuss several
alternatives of the expected Fisher information matrix when it is not
analytically available \cite[see also Section 4.4
of][]{livi:giro:2014}. In the numerical examples involving binomial
and Poisson SGLMMs, we observe that the PCMALA with an appropriate
pre-conditioning matrix performs favorably than the advanced
MMALA. Thus, in practice, it is worthwhile to construct suitable
PCMALA chains that may have superior performance than the modern
computationally expensive versions of MALA like the MMALA chain. On
the other hand, MMALA may dominate the PCMALA with the
pre-conditioning matrices used here for heavy-tailed distributions or
targets with a fast changing Hessian, for example, the perturbed
Gaussian density of \cite{chew:lu:ahn:2021} or the Example 4 of
\cite{gorh:dunc:2019} \cite[see also][]{tayl:2015,latu:robe:2011}.

Here, we have not considered a quantitative bound for the total
variation norm \eqref{eq:ge}, although with some modification of our
results such bounds can be obtained. For example, \cite{rose:1995} use
the method of coupling along with the drift and minorization technique
to construct such quantitative bounds. On the other hand, these bounds
are often too conservative to be used in practice
\citep{qin:hobe:2021}. Recently, \cite{durm:moul:2015} and
\cite{durm:moul:2019} build some quantitative bounds for certain MALA
and ULA chains. We believe that our results are a useful pre-cursor to
constructing sharper quantitative bounds for position dependent MALA
chains.

As mentioned before, \cite{livi:beta:byrn:2019} establish geometric
ergodicity of the HMC when the `mass matrix' in the `kinetic energy'
is fixed \cite[see also][]{mang:smit:2021}. On the other
hand, \cite{giro:cald:2011} argue that a position-dependent mass
matrix in the HMC may be preferred, and they develop the Riemann
manifold HMC (RMHMC). The techniques of this paper can be extended to
establish convergence results of the RMHMC algorithms and we plan to
undertake this as a future study. Finally, Langevin methods have been
applied to several Bayesian models \cite[see
e.g.][]{moll:syve:waag:1998, giro:cald:2011, neal:2012}. It would be
interesting to compare the performance of the PMALA and the PCMALA in
the context of these examples.

\begingroup
\small
\setstretch{0.9}
\bibliographystyle{asadoi}
\bibliography{ref}
\endgroup

\input{Appendix.tex}

\end{document}

%% file: Appendix.tex
\pagebreak
\begin{center}
\textbf{\large Supplement to \\
``Convergence of position-dependent MALA with application to conditional simulation in GLMMs" \\Vivekananda Roy and Lijin Zhang}
\end{center}
\setcounter{equation}{0}
\setcounter{figure}{0}
\setcounter{table}{0}
\setcounter{page}{1}
\setcounter{section}{0}
\makeatletter
\renewcommand{\thesection}{S\arabic{section}}
\renewcommand{\thesubsection}{\thesection.\arabic{subsection}}
\renewcommand{\theequation}{S\arabic{equation}}
\renewcommand{\thefigure}{S\arabic{figure}}
\renewcommand{\bibnumfmt}[1]{[S#1]}
\renewcommand{\citenumfont}[1]{S#1}
\renewcommand{\thetable}{S\arabic{table}}

  \section{Proofs of results}
\label{sec:approof}
\begin{proof}[Proof of Theorem~\ref{thm:gmalage}]
  From the form of \eqref{eq:propgmala}, and by A2, we know that $P$
  is $\phi$-irreducible and aperiodic. Let $C$ be a nonempty compact
  set. Since $f$ is bounded away from $0$ and $\infty$ on compact
  sets and A1 and A2 are in force, we have
  $\varepsilon =\inf_{x, y \in C} q(x, y) >0$ and
  $u = \sup_{x \in C} f(x) < \infty$. Let $B \subseteq C$. Then for
  any $x \in C$
  \begin{align*}
    P(x, B) &\ge \int_{A(x) \cap B} q(x, y) \alpha(x, y) dy + \int_{R(x) \cap B} q(x, y) \alpha(x, y) dy\\
            &= \int_{A(x) \cap B} q(x, y) dy + \int_{R(x) \cap B} \frac{f(y)}{f(x)} q(y, x)  dy\\
    &\ge \varepsilon \int_{A(x) \cap B} \frac{f(y)}{u} dy + \frac{\varepsilon}{u} \int_{R(x) \cap B} f(y) dy = \frac{\varepsilon}{u} F(B).
  \end{align*}
Thus, $C$ is small. Let
  $V_s(x)=\exp\{s\|G_2^{-1/2} x\|\}$, with $s>0$. We will show that
  with this drift function, Proposition~\ref{prop:mhge} holds, implying
  geometric ergodicity of the MH chain.  From \eqref{eq:mhtr} and
  \eqref{eq:rx}, we have
\[
    PV_s(x) = \int_{R^d}q(x,y)\alpha(x,y)V_s(y) dy+V_s(x)\int_{R(x)}q(x,y)(1-\alpha(x,y))dy,
\]
implying
\begin{equation}
  \label{eq:drmh}
  \frac{PV_s(x)}{V_s(x)} \leq \int_{R^d}q(x,y)\frac{V_s(y)}{V_s(x)} dy + \int_{R(x)}q(x,y)(1-\alpha(x,y))dy.
\end{equation}
Since $V_s(x)=\exp\{s\|G_2^{-1/2} x\|\}$, by A1, the first term in the right side of \eqref{eq:drmh} is as large as
\begin{equation}
   \label{eq:drmh1}
   (2\pi h)^{-d/2}|G_1|^{-1/2} \int_{R^d} \exp\Big\{-\frac{1}{2h}(y-c(x))^{\top}G_2^{-1}(y-c(x))+s(\|G_2^{-1/2} y\|-\|G_2^{-1/2} x\|)\Big\}dy. 
  \end{equation}
Since
$\|G_2^{-1/2}y\| \le \|G_2^{-1/2}(y-c(x))\| + \|G_2^{-1/2} c(x)\|$,
letting $z= G_2^{-1/2}(y-c(x))$, from \eqref{eq:drmh1}, it follows that the first term in the right side of \eqref{eq:drmh} is as large as
\begin{equation}
\label{eq:drmh2}
 \frac{\exp\{-s(\|G_2^{-1/2} x\|-\|G_2^{-1/2}c(x)\| -sh/2)\}}{(2\pi h)^{d/2}|G_1|^{1/2}|G_2|^{-1/2}}\int_{R^d}\exp\big\{-[\|z\|^2-2sh\|z\| + s^2h^2]/2h\big\}dz. 
\end{equation}
Now, we consider the polar transformation
$(z_1, \dots, z_d) \rightarrow (r, \theta_1, \theta_2, \dots,
\theta_{d-1})$ such that
$z_1=r\cos \theta_1, z_{2} =r\sin\theta_1\cos\theta_2,\dots,
z_{d-1}=r\sin\theta_1 \dots \sin\theta_{d-2}\cos\theta_{d-1},
z_{d}=r\sin\theta_1 \dots \sin\theta_{d-2}\sin\theta_{d-1}$.  Here,
$r >0, 0<\theta_{d-1}<2\pi$, $0<\theta_i<\pi, i=1,\dots,d-2$, and
the Jacobian is $r^{d-1} \prod_{i=1}^{d-2} \sin^{d-1-i}\theta_i$. Thus,
\begin{align}
  \label{eq:drmh3}
   \int_{R^d}\exp\Bigg\{-\frac{\|z\|^2-2sh\|z\| + s^2h^2}{2h}\Bigg\}dz &\le \int_{0}^{2\pi}\int_{0}^{\pi}\dots\int_{0}^{\pi}\int_{0}^{\infty}\exp\Bigg\{\frac{-(r-hs)^2}{2h}\Bigg\}r^{d-1}drd\theta_1\dots d\theta_{d-2}d\theta_{d-1}\nonumber\\
 &= 2\pi^{d-1}\int_{0}^{\infty}\exp\{-(r-hs)^2/(2h)\}r^{d-1}dr.
\end{align}
Using \eqref{eq:drmh2} and \eqref{eq:drmh3}, from \eqref{eq:drmh} we have
\begin{equation}
  \label{eq:pv}
  \limsup_{\|x\|\rightarrow\infty} \frac{PV_s(x)}{V_s(x)} \leq C_2(s) \exp(-s\eta)+ C_1.
\end{equation}
Thus under A4, \eqref{eq:limsupp} holds. Also, from \eqref{eq:drmh}--\eqref{eq:drmh3} by A2 we have
\[
  \sup_{\|x\| \le k} \frac{PV_s(x)}{V_s(x)} \le 1 + C_2(s) \sup_{\|x\| \le k} \exp\{s(\|G_2^{-1/2}c(x)\| -\|G_2^{-1/2} x\|)\} < \infty.
\]
Hence, the proof follows from Proposition~\ref{prop:mhge}.
\end{proof}

\begin{proof}[Proof of Theorem~\ref{thm:gmalage2}]
As in the proof of Theorem~\ref{thm:gmalage}, we know that $P$
  is $\phi$-irreducible, aperiodic, and nonempty compact
  sets are small.  Let $V(x)=x^{\top}x$. Then,
  \begin{align}
    \label{eq:normquad}
    \int_{R^d}\frac{V(y)}{V(x)}q(x,y) dy &\le \frac{1}{x^{\top}x} (2\pi h)^{-d/2}|G_1|^{-1/2} \int_{R^d} \big[y^{\top}y\big] \exp\Big\{-\frac{1}{2h}(y-c(x))^{\top}G_2^{-1}(y-c(x))\Big\}dy\nonumber\\ &= \frac{c(x)^{\top}c(x) + tr(hG_2)}{x^{\top}x} (|G_2|/|G_1|)^{1/2}. 
  \end{align}
  Thus, from \eqref{eq:drmh} and \eqref{eq:normquad} we have
\begin{equation}
  \label{eq:pv2}
  \limsup_{\|x\|\rightarrow\infty} \frac{PV(x)}{V(x)} \leq \Big(\frac{|G_2|}{|G_1|}\Big)^{1/2}\limsup_{\|x\|\rightarrow\infty}\frac{\|c(x)\|^2}{\|x\|^2} + C_1.
\end{equation}
Also, note that by A2 we have
\[
  \sup_{\|x\| \le k} PV(x) \le k^2 + (|G_2|/|G_1|)^{1/2} \sup_{\|x\| \le k} \Big[c(x)^{\top}c(x) + tr(hG_2)\Big]  < \infty.
\]
  Hence, the proof follows by \eqref{eq:pv2} and applying Proposition~\ref{prop:mhge} on the function $V(x) +1$ as A5 is in force.  
\end{proof}

\begin{proof}[Proof of Theorem~\ref{thm:necmgf}]
  Since $c(x) = x + h e(x)$, from \eqref{eq:propgmala} we have
  \begin{equation}
    \label{eq:necmgf}
    Q(x, B_k(x)) \ge (2\pi h)^{-d/2}|G_2|^{-1/2} \int_{\|z\| <k} \exp\Big\{-\frac{1}{2h}(z-he(x))^{\top}G_1^{-1}(z-he(x))\Big\}dz.
  \end{equation}
  Then, by $\|e(x)\| < M$ it follows that for given $\varepsilon >0$, there
  exists $k$ such that $Q(x, B_k(x))> 1-\varepsilon$. Thus, the result
  follows from Proposition~\ref{prop:mhnec}.
\end{proof}
\begin{proof}[Proof of Theorem~\ref{thm:gmalanclt}]
Choose $T>2/h$, such that when $S_1$ is large enough,
\begin{align*}
    \inf_{\|x\|>S_1}\frac{\|e(x)\|}{\|x\|}>T.
\end{align*}
Define $B^2_{k}(x)=\{y:\|y- c(x)\|\leq k \}$. By A1, for given $\varepsilon>0$, there exists $k_{\varepsilon}$, 
such that $\int_{R^d/B^2_{k_\varepsilon}(x)}q(x,y)dy<\varepsilon/2$. To
simplify notations, for the rest of this proof, we denote
$B^2_{k_\varepsilon}(x)$ by $B^2_{\varepsilon}(x)$. When
$y\in B^2_{\varepsilon}(x)$, $q(x,y)$ is bounded away from 0, as
\begin{equation}
\label{eq:qxylo}
      q(x,y)\geq  (2\pi h\zeta_{2}^+)^{-d/2}\exp\{-k_{\varepsilon}^2/[2h\zeta_{1+}]\}.
\end{equation}
Note that, the proposed $y$ is generated as
$y= x+h e(x)+\sqrt{h G(x)}\epsilon$, which is either accepted or
rejected with the chain staying at the current position $x$. Here
$\epsilon\sim N(0, I_d)$. Since,
\[
  \|y\|=\|x+h e (x)+\sqrt{h G(x)} \epsilon\|\geq h\|e(x)\|-\|x\|-\sqrt{h}\|\sqrt{G(x)}\epsilon\|,
\]
we have
\[
    \frac{\|y\|}{\|x\|}\geq h\frac{\|e(x)\|}{\|x\|}-1-\sqrt{h}\frac{\|\sqrt{G(x)}\epsilon\|}{\|x\|}.
\]
Hence, by \eqref{eq:gmalanclt}, $\exists \; S_2$ such that when $\|x\|> S_2$ and
$y\in B^2_{\varepsilon}(x)$, we have $\|y\|>\|x\|$. Let $S=$ max
$(S_1, S_2)$. Thus, when $y\in B_{\varepsilon}(x)$, and $\|x\|>S$,
\[
    \|he(y)\|>hT\|y\|>2\|y\|>\|x\|+\|y\|>\|x-y\|,
\]
and hence,
\[
    \|x-y-he(y)\|\geq\|he(y)\|-\|x-y\| \geq hT\|y\|-2\|y\|.
\]
So, for $y\in B_{\varepsilon}(x)$, and $\|x\|>S$, we have
\begin{align*}
q(y,x)&\leq (2\pi h \zeta_{1+})^{-d/2}\exp\big\{-\frac{1}{2h\zeta_{2}^+}\|x-y-he(y)\|^2\big\}\\
&\leq (2\pi h\zeta_{1+})^{-d/2}\exp\{-(hT-2)^2\|y\|^2/[2h\zeta_{2}^+]\},
\end{align*}
and, thus when $\|x\|\rightarrow\infty$,
\begin{equation}
  \label{eq:qyxup}
    \sup_{y\in B^2_{\varepsilon}(x)}q(y,x)\leq (2\pi h \zeta_{1+})^{-d/2}\exp\{-(hT-2)^2\|x\|^2/[2h\zeta_{2}^+]\}
    \rightarrow 0.  
\end{equation}
From \eqref{eq:qxylo} and \eqref{eq:qyxup} we have 
\begin{equation}
  \label{eq:qyxxy}
    \sup_{y\in B^2_{\varepsilon}(x)}\frac{q(y,x)}{q(x,y)} \le \frac{\sup_{y\in B^2_{\varepsilon}(x)}q(y,x)}{\inf_{y\in B^2_{\varepsilon}(x)}q(x,y)}
    \rightarrow 0 \;\text{as}\; \|x\|\rightarrow\infty.
\end{equation}
Starting with $\|x_0\|>S$ and $f(x_0)>0$, define
$x_n=\text{arg}\sup\{f(y);y\in B^2_{\varepsilon}(x_{n-1})\}$.  Note
that $\|x_n\|\rightarrow\infty$ when $n\rightarrow\infty$. 
Assume that the MH chain is GE, then, from Section~\ref{sec:mcbg} there exists $\varepsilon>0$, such that $\text{ess}\sup r(x)<1-\varepsilon$. Now,
\begin{align*}
    \text{ess}\sup r(x_n)&= \text{ess}\sup\big\{1-\int_{R^d}\alpha(x_n,y)q(x_n,y)dy\big\}\\
    &\geq 1-\text{ess}\sup\big\{\int_{R^d}\alpha(x_n,y)q(x_n,y)dy\big\}.
\end{align*}
Thus,
\begin{align*}
    1-\text{ess}\sup r(x_n)\leq \sup \int_{R^d}\alpha(x_n,y)q(x_n,y)dy,
\end{align*}
and with the fact that $\int_{R^d/B^2_{\varepsilon}(x_n)}q(x_n,y)dy\leq \varepsilon/2$, we have
\begin{align*}
    1-\text{ess}\sup r(x_n)&\leq\sup\int_{R^d/B^2_{\varepsilon}(x_n)}\Big\{1\wedge\frac{f(y)}{f(x_n)}\frac{q(y,x_n)}{q(x_n,y)} \Big\}q(x_n,y)dy\\&+\sup\int_{B^2_{\varepsilon}(x_n)}\Big\{1\wedge\frac{f(y)}{f(x_n)}\frac{q(y,x_n)}{q(x_n,y)} \Big\}q(x_n,y)dy\\
    &\leq \sup\int_{R^d/B^2_{\varepsilon}(x_n)}q(x_n,y)\mu(dy)+\sup\int_{B^2_{\varepsilon}(x_n)}\Big\{1\wedge\frac{f(y)}{f(x_n)}\frac{q(y,x_n)}{q(x_n,y)}\Big\} q(x_n,y)dy\\
    &\leq \frac{\varepsilon}{2}+\sup\int_{B^2_{\varepsilon}(x_n)}1\wedge\frac{f(y)}{f(x_n)}\frac{q(y,x_n)}{q(x_n,y)} q(x_n,y)dy.
\end{align*}
Thus,
\begin{align}
    \label{eq:piqyxxy}
   \varepsilon< &1-\text{ess}\sup r(x_n) \le \frac{\varepsilon}{2}+\sup\int_{B^2_{\varepsilon}(x_n)}\Big\{1\wedge\frac{f(y)}{f(x_n)}\frac{q(y,x_n)}{q(x_n,y)} \Big\}q(x_n,y)dy\nonumber\\
    &\Rightarrow \sup\int_{B^2_{\varepsilon}(x_n)}\Big\{1\wedge\frac{f(y)}{f(x_n)}\frac{q(y,x_n)}{q(x_n,y)}\Big\} q(x_n,y)\mu(dy)>\frac{\varepsilon}{2}.
\end{align}
From \eqref{eq:qyxxy}, when $n>N$ and $N$ is large enough, we have $\sup_{y\in B^2_{\varepsilon}(x_n)}[q(y,x_n)/q(x_n,y)]< \varepsilon/4$. Thus from \eqref{eq:piqyxxy} we have
\begin{align*}
    \frac{\varepsilon}{2} \le \sup\int_{B^2_{\varepsilon}(x_n)} \Big\{1\wedge \frac{\varepsilon}{4}\frac{f(x_{n+1})}{f(x_n)}\Big\}q(x_n,y)dy \leq \frac{\varepsilon}{4}\frac{f(x_{n+1})}{f(x_n)},
\end{align*}
implying $f(x_{n+1})> 2f(x_n)$, which contradicts that $f$ is bounded.
Therefore, the MH chain is not geometric ergodic. 
\end{proof}
\begin{proof}[Proof of Proposition~\ref{prop:pcula}]
    Since $c(x)$ is continuous, by Fatou's lemma for a fixed open set $A \in \mathbb{R}^d$,
    \[
      \liminf_{x_n \rightarrow x} Q(x_n, A) = \liminf_{x_n \rightarrow
        x} \frac{|G|^{-1/2}}{(2\pi h)^{d/2}} \int_{A}
      \exp\Big\{-\frac{1}{2h}(y-c(x_n))^{\top}G^{-1}(y-c(x_n))\Big\}dy
      \ge Q(x, A).\] Thus, $\{X_n\}_{n \ge 0}$ is a Feller chain. From
    the proof of Theorem~\ref{thm:gmalage}, we have
    \[
       \limsup_{\|x\|\rightarrow\infty} \frac{QV_s(x)}{V_s(x)} <1,
     \]
     and as $V_s$ is unbounded off compact sets, by
     Proposition~\ref{prop:gedrif}, $\{X_n\}_{n \ge 0}$ is
     GE when A4 holds. Similarly, the proof for A5 follows by
     Proposition~\ref{prop:gedrif}, and using the drift function $V$
     from Theorem~\ref{thm:gmalage2}.
   \end{proof}
   \begin{proof}[Proof of Theorem~\ref{thm:gmalabinom}]
  For the PCMALA chain A1 holds automatically. Recall that the
  proposal density for PCMALA is $N(x+hG\nabla\log f(x)/2, hG)$. Thus,
  from \eqref{eq:binom12der} it follows that the proposal density for
  the PCMALA is \eqref{eq:propgmala} with $G(x) = G$ and
  \begin{equation}
    \label{eq:spmean}
     c(x)=x+(h/2) G(-\Sigma^{-1}x+b(x)), 
  \end{equation}
where $b(x)=z-\ell\cdot(e^x/[1+e^x])+\Sigma^{-1}D\beta$
is bounded. Thus, A2 holds for the PCMALA chain.   We now show that A3 holds for the PCMALA chain.
For a given $0< \varepsilon < 1$, set $B_{\varepsilon}^3(x)=\{y: \|G^{-1/2}(y-c(x))\|<k_{\varepsilon}\}$, and $k_\varepsilon>0$, such that $\int_{R^{m}\setminus B_{\varepsilon}^3(x)}q(x,y)dy<\varepsilon$.

From \eqref{eq:post} the acceptance probability $\alpha (x,y)$ in \eqref{eq:accpro} becomes 
\begin{align*}
  1\wedge \exp\{-\log q(x,y)+\log q(y,x)+\log p(y) -\log p(x) + \sum_{i=1}^m\log a(z_i,\mu_i(y))-\sum_{i=1}^m\log a(z_i,\mu_i(x))\}. 
\end{align*}
We will show that the proposal $y \in B_{\varepsilon}^3(x)$ is always accepted when $\|x\|\rightarrow \infty$. 
From \eqref{eq:propgmala} and \eqref{eq:spmean} for the PCMALA we have
\begin{align*}
    &-\log q(x,y) -m \log(2h\pi)/2 - \log|G|/2\\
    =&\frac{1}{2h}\Big(y-x-\frac{h}{2}G(-\Sigma^{-1}x+b(x))\Big)^\top G^{-1}\Big(y-x-\frac{h}{2}G(-\Sigma^{-1}x+b(x))\Big)\\
    =&\frac{1}{2h}y^\top G^{-1}y+\frac{1}{2h}x^\top G^{-1}x+\frac{h}{8}x^\top\Sigma^{-1}G\Sigma^{-1}x+\frac{h}{8}b(x)^\top Gb(x)\\
    &-\frac{1}{h}y^\top G^{-1}x+\frac{1}{2}y^\top\Sigma^{-1}x-\frac{1}{2}y^\top b(x)-\frac{1}{2}x^\top\Sigma^{-1}x+\frac{1}{2}x^\top b(x)-\frac{h}{4}x^\top\Sigma^{-1}Gb(x).
\end{align*}
Let 
\begin{align}
  \label{eq:defjs}
 L_1&=\frac{h}{8}(x^\top\Sigma^{-1}G\Sigma^{-1}x-y^\top\Sigma^{-1}G\Sigma^{-1}y),\nonumber\\
 L_2&=\frac{h}{8}(b(x)^\top Gb(x)-b(y)^\top Gb(y)),\nonumber\\
 L_3&=\frac{1}{2}x^\top(I-h\Sigma^{-1}G/2)b(x)-\frac{1}{2}y^\top(I-h\Sigma^{-1}G/2)b(y) +\frac{1}{2}x^\top b(y)-\frac{1}{2}y^\top b(x),\;\text{and} \nonumber\\
 L_4&=\log p(y) -\log p(x) + \sum_{i=1}^m\log a(z_i,\mu_i(y))-\sum_{i=1}^m\log a(z_i,\mu_i(x))-\frac{1}{2}x^\top\Sigma^{-1}x+\frac{1}{2}y^\top\Sigma^{-1}y.
\end{align}
Note that $\alpha(x,y)=1\wedge \exp\{L_1+L_2+L_3+L_4\}$. If
$L_1+L_2+L_3+L_4\geq 0$, then the proposed $y$ would always be
accepted.  Again, from \eqref{eq:spmean}, for $y \in B_\varepsilon^3(x)$, we have
\begin{align}
  \label{eq:yquad}
    y^\top\Sigma^{-1}G\Sigma^{-1}y&=\big(\big(I-\frac{h}{2}G\Sigma^{-1}\big)x+O(1)\big)^\top\Sigma^{-1}G\Sigma^{-1}\big(\big(I-\frac{h}{2}G\Sigma^{-1}\big)x+O(1)\big) \nonumber\\
    &=x^\top\Sigma^{-1}G\Sigma^{-1}x+\frac{h^2}{4}x^\top\Sigma^{-1}G\Sigma^{-1}G\Sigma^{-1}G\Sigma^{-1}x+
    O(1)^\top\Sigma^{-1}G\Sigma^{-1}O(1) \nonumber \\ & -hx^\top\Sigma^{-1}G\Sigma^{-1}G\Sigma^{-1}x+
    2x^\top\Sigma^{-1}G\Sigma^{-1}O(1)-hx^\top\Sigma^{-1}G\Sigma^{-1}G\Sigma^{-1}O(1).
\end{align}
Let $\psi_{+} (\zeta_{+})$ and $\psi^{+} (\zeta^{+})$ be the smallest and the largest eigenvalue of $\Sigma^{-1} (G)$, respectively.
Note that, $x^\top\Sigma^{-1}G\Sigma^{-1}x/\|x\|^2 \in [\psi_+^2\zeta_+,\psi^{+2}\zeta^+]$. Similarly
$x^\top\Sigma^{-1}G\Sigma^{-1}G\Sigma^{-1}x/\|x\|^2 \in [\psi_+^3\zeta_+^2,\psi^{+3}\zeta^{+2}],$ and 
 $x^\top\Sigma^{-1}G\Sigma^{-1}G\Sigma^{-1}G\Sigma^{-1}x/\|x\|^2 \in [\psi_+^4\zeta_+^3,\psi^{+4}\zeta^{+3}]$. Thus, from \eqref{eq:defjs} and \eqref{eq:yquad} we have
\begin{align*}
    L_1=&\frac{h}{8}\big(-\frac{h^2}{4}x^\top\Sigma^{-1}G\Sigma^{-1}G\Sigma^{-1}x-O(1)^\top\Sigma^{-1}O(1)+ hx^\top\Sigma^{-1}G\Sigma^{-1}x-2x^\top\Sigma^{-1}O(1)\\&+hx^\top\Sigma^{-1}G\Sigma^{-1}O(1)\big)\\
    \geq&-\frac{h^3}{32}\psi^{+3}\zeta^{+2}\|x\|^2+\frac{h^2}{8}\psi_+^2\zeta_+\|x\|^2+o(\|x\|^2).
\end{align*}
So, if $h\in (0, 4\psi_+^2\zeta_+/[\psi^{+3}\zeta^{+2}]),$ when
$\|x\|\rightarrow\infty$, $L_1\rightarrow\infty$. Also, note that for
such $h,$ $L_1 \sim \|x\|^2$, $L_2$ is bounded and $L_3=o(L_1)$. Here,
$i(x) \sim j(x)$ means $i(x)/j(x) \rightarrow c$ for some constant
$c>0$. Since $\lim_{t \rightarrow - \infty}\log (1+\exp{[t]})/t^2 = 0$
and by L'Hospital's rule,
$\lim_{t \rightarrow\infty}\log (1+\exp{[t]})/t^2 = 0$, we have
\begin{align*} 
   L_4=&\log p(y) -\log p(x) + \sum_{i=1}^m\log a(z_i,\mu_i(y))-\sum_{i=1}^m\log a(z_i,\mu_i(x))-\frac{1}{2}x^\top\Sigma^{-1}x+\frac{1}{2}y^\top\Sigma^{-1}y\\
   =&-\frac{1}{2}(y-D\beta)^\top\Sigma^{-1}(y-D\beta)+\sum_{i=1}^m \{z_i y^{(i)}-\ell_i\log(1+\exp(y^{(i)}))\}+\\
   &\frac{1}{2}(x-D\beta)^\top\Sigma^{-1}(x-D\beta) - \sum_{i=1}^m \{z_i x^{(i)}-\ell_i\log(1+\exp(x^{(i)}))\}-\frac{1}{2}x^\top\Sigma^{-1}x+\frac{1}{2}y^\top\Sigma^{-1}y\\
   =&o(\|x\|^2).
\end{align*}
Therefore, for large $\|x\|$, $B_{\varepsilon}^3(x)\subseteq A(x)$. Recall that,
$A(x)= \{y:f(x)q(x,y)\leq f(y)q(y,x)\} = R(x)^c$. Thus,
\begin{align*}
    \liminf_{\|x\|\rightarrow\infty}\int_{A(x)}q(x,y)dy&\geq\liminf_{\|x\|\rightarrow\infty}\int_{B_{\varepsilon}^3(x)}q(x,y)dy >1-\varepsilon>0.
\end{align*}
Hence,
\begin{align*}
  C_1=\limsup_{\|x\|\rightarrow\infty}\int_{R(x)}q(x,y)(1-\alpha(x,y))dy &\leq\limsup_{\|x\|\rightarrow\infty}\int_{R(x)}q(x,y)dy \\&\leq 1-\liminf_{\|x\|\rightarrow\infty}\int_{A(x)}q(x,y)dy<1.
 \end{align*}
 Thus, A3 holds for the PCMALA chain. Next, we verify A4.  From
 \eqref{eq:spmean}, note that
 \begin{align}
   \label{eq:ineqspm}
   \|x\|-\|c(x)\|
   &=\|x\|-\|(I-(h/2)G\Sigma^{-1})x+ (h/2)G b(x)\| \nonumber\\
   &\geq \|x\|-\|(I-(h/2)G\Sigma^{-1})x\|-\|(h/2)G b(x)\|.
\end{align}
Now, if
$1-h\psi_+ \zeta_+ + h^2\psi^{+2} \zeta^{+2}/4<1 \Leftrightarrow h<4
\psi_+ \zeta_+/\psi^{+2}\zeta^{+2}$, then
$\|(I-(h/2)G\Sigma^{-1})x\|^2/x^\top x <1$. Since $b(x)$ is bounded,
and $\psi_+/\psi^{+} <1$, if
$h\in (0,4\psi_{+}^2\zeta_{+}/(\psi^{+3}\zeta^{+2}))$, from
\eqref{eq:ineqspm}, we have
$\liminf_{\|x\|\rightarrow\infty}\big(\|x\|-\|c(x)\|\big)=\infty$. Thus, by
Remark~\ref{rem:gmalage}, it follows that A4 holds for the PCMALA
chain. Thus for $h\in (0,4\psi_{+}^2\zeta_{+}/(\psi^{+3}\zeta^{+2}))$ geometric ergodicity of the PCMALA chain follows from Theorem~\ref{thm:gmalage}.

Next, we verify A1--A4 for the MMALA chain. 
Since
\begin{equation}
  \label{eq:ibd}
 G_1\equiv (0.25\text{diag} (\ell)+
  \Sigma^{-1})^{-1} \le \mathscr{I}^{-1}(x) \le \Sigma \equiv G_2,
\end{equation}
A1 holds
for the MMALA. The mean of the proposal distribution for the MMALA is
$c(x)=x+(h/2)\mathscr{I}^{-1}(x)\nabla\log f(x)+(h/2)\Gamma(x)$ where
$\Gamma(x)$ is given in \eqref{eq:gambin}. Thus,
  \begin{align}
    \label{eq:mmalamean}
    c(x)=x+(h/2)\mathscr{I}^{-1}(x)(-\Sigma^{-1}x + \kappa(x)),
\end{align}
where $\kappa(x)=b(x) + \mathscr{I}(x)\Gamma(x)$. 
From \eqref{eq:binom3der} it follows that
$(\partial \mathscr{I}_{jj} (x)/\partial x_{j})$ is bounded.  By
\eqref{eq:ibd} we have $\mathscr{I}^{-1}_{jj} (x)$ is bounded for all
$j$, and then, an application of the Cauchy-Schwartz inequality shows that
$\mathscr{I}^{-1}_{ij} (x)$ is bounded. Thus, from \eqref{eq:gambin} it
follows that $\kappa(x)$ is bounded, and A2 holds
for the MMALA chain.
  

Note that,
\begin{align*}
    &-\log q(x,y) -m \log(2h\pi)/2 + \log|\mathscr{I}(x)|/2\\
    =&\frac{1}{2h}\big(y-x-\frac{h}{2}\big\{\mathscr{I}^{-1}(x)(-\Sigma^{-1}x+\kappa(x))\big\}\big)^\top \mathscr{I}(x)\big(y-x-\frac{h}{2}\big\{\mathscr{I}^{-1}(x)(-\Sigma^{-1}x+\kappa(x))\big\}\big)\\
    =&\frac{1}{2h}y^\top \mathscr{I}(x) y+\frac{1}{2h}x^\top \mathscr{I}(x) x+\frac{h}{8}x^\top \Sigma^{-1}\mathscr{I}^{-1}(x)\Sigma^{-1}x+\frac{h}{8}\kappa(x)^\top \mathscr{I}^{-1}(x) \kappa(x) -\frac{1}{h}y^\top \mathscr{I}(x) x\\&+\frac{1}{2}y^\top \Sigma^{-1} x-\frac{1}{2}y^\top \kappa(x)-\frac{1}{2}x^\top \Sigma^{-1}x+\frac{1}{2}x^\top \kappa(x)-\frac{h}{4}x^\top\Sigma^{-1}\mathscr{I}^{-1}(x)\kappa(x).
\end{align*}
Let
\begin{align*}
 L'_1& = h(x^\top\Sigma^{-1}\mathscr{I}^{-1}(x) \Sigma^{-1}x-y^\top\Sigma^{-1}\mathscr{I}^{-1}(y) \Sigma^{-1}y)/8,\\ 
 L'_2&=h(\kappa(x)^\top \mathscr{I}^{-1}(x) \kappa(x)-\kappa(y)^\top \mathscr{I}^{-1}(y)\kappa(y))/8 + \log(|\mathscr{I}(y)|/|\mathscr{I}(x)|)/2,\\
 L'_3&=\frac{1}{2}x^\top(I-(h/2)\Sigma^{-1}\mathscr{I}^{-1}(x))\kappa(x)-\frac{1}{2}y^\top(I-(h/2)\Sigma^{-1}\mathscr{I}^{-1}(y))\kappa(y) +\frac{1}{2}x^\top \kappa(y)-\frac{1}{2}y^\top \kappa(x),\\
 L'_5 &= ([\{x^\top \mathscr{I}(x) x-x^\top \mathscr{I}(y) x\}+\{y^\top \mathscr{I}(x) y-y^\top \mathscr{I}(y)y\}]/2 -[x^\top \mathscr{I}(x) y-x^\top \mathscr{I}(y) y])/h.
\end{align*}
Note that $\alpha(x,y)=1\wedge \exp\{L'_1+L'_2+L'_3+L_4+L'_5\}$.  Define, $B^4_{k}(x) \equiv \{y:\|y- c(x)\|\leq k \; \& \; \prod_{i=1}^m x^{(i)}(x^{(i)} - y^{(i)}) >0\}$.  
We can find $\varepsilon>0$ and $k_\varepsilon$ 
such that
$\int_{R^m/B^4_{k_\varepsilon}(x)}q(x,y)dy<\varepsilon$. Thus, from \eqref{eq:ibd} and \eqref{eq:mmalamean}, for $y \in B_\varepsilon^4(x)$, we have
\begin{align}
  \label{eq:lonepr}
   &  y^\top \Sigma^{-1}\mathscr{I}^{-1}(y) \Sigma^{-1}y\\=&\big(\big(I-(h/2)\mathscr{I}^{-1}(x)\Sigma^{-1}\big)x+O(1)\big)^\top \Sigma^{-1}\mathscr{I}^{-1}(y) \Sigma^{-1} \big(\big(I-(h/2)\mathscr{I}^{-1} (x) \Sigma^{-1}\big)x+O(1)\big)\nonumber\\
     =&x^\top \Sigma^{-1}\mathscr{I}^{-1}(y) \Sigma^{-1}x+(h^2/4)x^\top\Sigma^{-1}\mathscr{I}^{-1}(x)\Sigma^{-1}\mathscr{I}^{-1}(y) \Sigma^{-1} \mathscr{I}^{-1}(x)\Sigma^{-1}x\nonumber\\ & -hx^\top\Sigma^{-1}\mathscr{I}^{-1}(y) \Sigma^{-1} \mathscr{I}^{-1}(x)\Sigma^{-1}x+
     o(\|x\|^2).
 \end{align}
 Let $\iota (x^{(i)}) = e^{x^{(i)}}/(1+e^{x^{(i)}})$, $\omega(x^{(i)})= \{\iota (x^{(i)}) - \iota^2(x^{(i)})\}$, and
 $E(x) = \mbox{diag} (\ell_1\omega(x^{(1)}), \dots,\\ \ell_m\omega(x^{(m)})).$
 Note that
 $\omega'(t) = \iota(t)- 3\iota^2(t)+ 2\iota^3(t) =\iota(t)(1-
 \iota(t))(1- 2\iota(t))\gtreqless 0 \iff t \lesseqgtr 0$. That is,
 $\omega(t)$ is decreasing (increasing) on the positive (negative)
 half line. So, for $y \in B_\varepsilon^4(x)$, $ E(y) \ge E(x)$, implying
 $\mathscr{I}^{-1}(x) = (E(x) + \Sigma^{-1})^{-1} \ge (E(y) +
 \Sigma^{-1})^{-1} = \mathscr{I}^{-1}(y)$. So,
 $ x^\top \Sigma^{-1}\mathscr{I}^{-1}(x) \Sigma^{-1}x - x^\top
 \Sigma^{-1}\mathscr{I}^{-1}(y) \Sigma^{-1}x \ge 0$.
 Thus, from \eqref{eq:lonepr} we have
 \begin{align*}
   L'_1\ge&\frac{h}{8}(-h^2x^\top\Sigma^{-1}\mathscr{I}^{-1}(x)\Sigma^{-1}\mathscr{I}^{-1}(y) \Sigma^{-1} \mathscr{I}^{-1}(x)\Sigma^{-1}x/4 + hx^\top\Sigma^{-1}\mathscr{I}^{-1}(y) \Sigma^{-1} \mathscr{I}^{-1}(x)\Sigma^{-1}x) +
     o(\|x\|^2)\\
   \geq &\frac{h}{8}(h\psi_+^3 \zeta_{1 +}^{2}\|x\|^2-h^2\psi^{+4}\zeta_{2}^{+3}\|x\|^2/4)+
     o(\|x\|^2).
 \end{align*}
 Recall that, $\zeta_{i+}$ and $\zeta_{i}^{+}$ are the smallest and the largest
 eigenvalue of $G_i$, respectively for $i=1,2$. So, if
 $h\in (0, 4\psi_+^3 \zeta_{1 +}^{2}/[\psi^{+4}\zeta_{2}^{+3}]),$ when
 $\|x\|\rightarrow\infty$, $L'_1\rightarrow\infty$. 
 Also, $L'_2$ is bounded, and $L'_3 = o(L'_1)$.
Next, we consider $L'_5$. Note that $\mathscr{I}_{ij}(x)-\mathscr{I}_{ij}(y) = 0$ for $i \neq j$. Thus,
 \begin{align*}
   x^\top \mathscr{I}(x) x-x^\top \mathscr{I}(y) x = \sum_{i=1}^m (\mathscr{I}_{ii}(x)-\mathscr{I}_{ii}(y)) x^{(i)2} = \sum_{i=1}^m \ell_i [\omega(x^{(i)}) -\omega(y^{(i)})]x^{(i)2},
 \end{align*}
 and
 \[
   L'_5 = \sum_{i=1}^m \ell_i [\omega(x^{(i)})- \omega(y^{(i)})](x^{(i)} - y^{(i)})^2/2.
 \]
 By \eqref{eq:ibd} and \eqref{eq:mmalamean}, as
 $y \in B_\varepsilon^4(x)$, for small $h$ we have
 $\sum_{i=1}^m (x^{(i)} - y^{(i)})^2 = O(1)$. Since $\iota (t) \in (0, 1)$, for $y \in B_\varepsilon^4(x)$ we have $L'_5 = O(1)$. 
 Then, using similar
 arguments as in the proof of geometric ergodicity for the PCMALA chain, we can show that A3 holds for the
 MMALA chain. Finally, for the MMALA
 \[
    \|x\|-\|c(x)\|\geq \|x\|-\|(I-(h/2)\mathscr{I}^{-1}(x)\Sigma^{-1})x\|-\|(h/2)\mathscr{I}^{-1}(x) \kappa(x)\|.
\]
If
$ h<4\psi_+ \zeta_{1+}/\psi^{+2}\zeta_2^{+2}$, then
$\|(I-(h/2)\mathscr{I}^{-1}(x) \Sigma^{-1})x\|^2/x^\top x <1$. Since $\zeta_{1+}/\zeta_2^{+}<1$,  if
$h\in (0, 4\psi_+^3 \zeta_{1 +}^{2}/[\psi^{+4}\zeta_{2}^{+3}]),$ then A4 holds for the MMALA
chain. Thus geometric ergodicity of the MMALA follows from Theorem~\ref{thm:gmalage}.

Next, we consider the PCULA chain. From \eqref{eq:spmean} we know that the PCULA chain is given by
   \[
     X_n=(I-(h/2)G\Sigma^{-1})X_{n-1}+b(X_{n-1}) + \sqrt{h}G^{1/2}\epsilon, 
  \]
  where $b(x)=z-\ell\cdot(e^x/[1+e^x])+\Sigma^{-1}D\beta$. So, if $h$
  is chosen such that the singular values of $(I-(h/2)G\Sigma^{-1})$
  are less than one, then geometric ergodicity of the PCULA follows from
  \eqref{eq:pula_extend} and Remark~\ref{rem:pculasing}.
\end{proof}
\begin{proof}[Proof of Proposition~\ref{thm:pcmalapois}]
 Since
 $\lim_{\|x\|\rightarrow\infty}\prod_{i=1}^m
 \exp\{z_ix^{(i)}\}\exp\{-\exp\{x^{(i)}\}\}=0,$ if $z_i >0 \; \forall i$, and
 otherwise is bounded, it follows that the target density $f(x)$ is
 bounded. From \eqref{eq:poisder} we have
 \begin{equation}
   \label{eq:poisgl}
   \|\nabla\log f(x)\| = \|z-\exp\{x\}-\Sigma^{-1}(x-D\beta)\| \ge \|\exp\{x\}\|-\|z\|-\|\Sigma^{-1}(x-D\beta)\|.
 \end{equation}
 Next, for $x^{(i)} >0, i=1,\dots,m$, from \eqref{eq:poisgl} we have
 \begin{align*}
   \liminf_{\|x\| \rightarrow\infty}   \|G\nabla\log f(x)\|/\|x\| &\ge \liminf_{\|x\|\rightarrow\infty} \zeta_{+}\{1-(\|z\|+\|\Sigma^{-1}(x-D\beta)\|)/\|\exp\{x\} \|\}\|\exp\{x\}\|/\|x\|\\& = \infty.
 \end{align*}
Hence, the result follows from Theorem~\ref{thm:gmalanclt}.
\end{proof}

\section{Additional numerical results for the SGLMMs}
\label{sec:apppois}
In this section, we include some tables and figures from the analysis of simulated data from the SGLMMs.

\vspace*{-.1in}

\begin{table}[H]
\caption{ESS values for the MH chains for the Poisson SGLMM with the log link}
\centering
\begin{tabular}{c c c c  c} 
 \hline\hline
Algorithm & $G$ matrix &ESS(1, 175, 350) & ESS/min& mESS \\  [0.5ex]
 \hline\hline
  \multirow{4}{4em}{RWM}& $I$ & ( 8,19,36 ) &( 0.03,0.07,0.13 )  &  1,052 \\ 
& $\Sigma$ & ( 12,13,20 ) &( 0.05,0.05,0.08 ) &  1,052 \\
& diag $\hat{\mathcal{I}}^{-1}$ & ( 16,25,20 ) & ( 0.06,0.10,0.08 ) &  1,045\\
& $\hat{\mathcal{I}}^{-1}$ & ( 14,19,9 ) & ( 0.06,0.08,0.04 ) & 1,053 \\
\hline
\multirow{4}{4em}{PCMALA}& I & ( 6,6,8 ) & ( 0.02,0.03,0.03 ) &  1,026 \\ 
& $\Sigma$&( 7,6,8 ) &( 0.03,0.03,0.03 ) &  1,034\\
& diag $\hat{\mathcal{I}}^{-1}$ &( 37,37,25 )&( 0.17,0.17,0.12 )& 1,055\\
& $\hat{\mathcal{I}}^{-1}$ &( 7,764,8,667,8,138 ) &( 36.95,41.25,38.73 ) &   12,535\\
PMALA&  & ( 132,226,160 ) &( 0.46,0.78,0.55 )  &  1,168\\
\hline\hline
\end{tabular}
\label{table:ess.pl}
\end{table}


\begin{table}[H]
\caption{MSJD values for the MH chains for the Poisson SGLMM with the log link}
\centering
\begin{tabular}{c c c c c c c c c} 
 \hline\hline
RWM1 &RWM2 &RWM3 & RWM4&PCMALA1&PCMALA2&PCMALA3&PCMALA4&PMALA\\  
 \hline\hline
  0.018&0.023 &0.032&0.013&4.52e-05& 5.18e-09&0.049&11.70& 0.222\\
\hline\hline
\end{tabular}
\label{table:msejd.pl}
\end{table}

\begin{figure*}[t]
  \includegraphics[width=\linewidth]{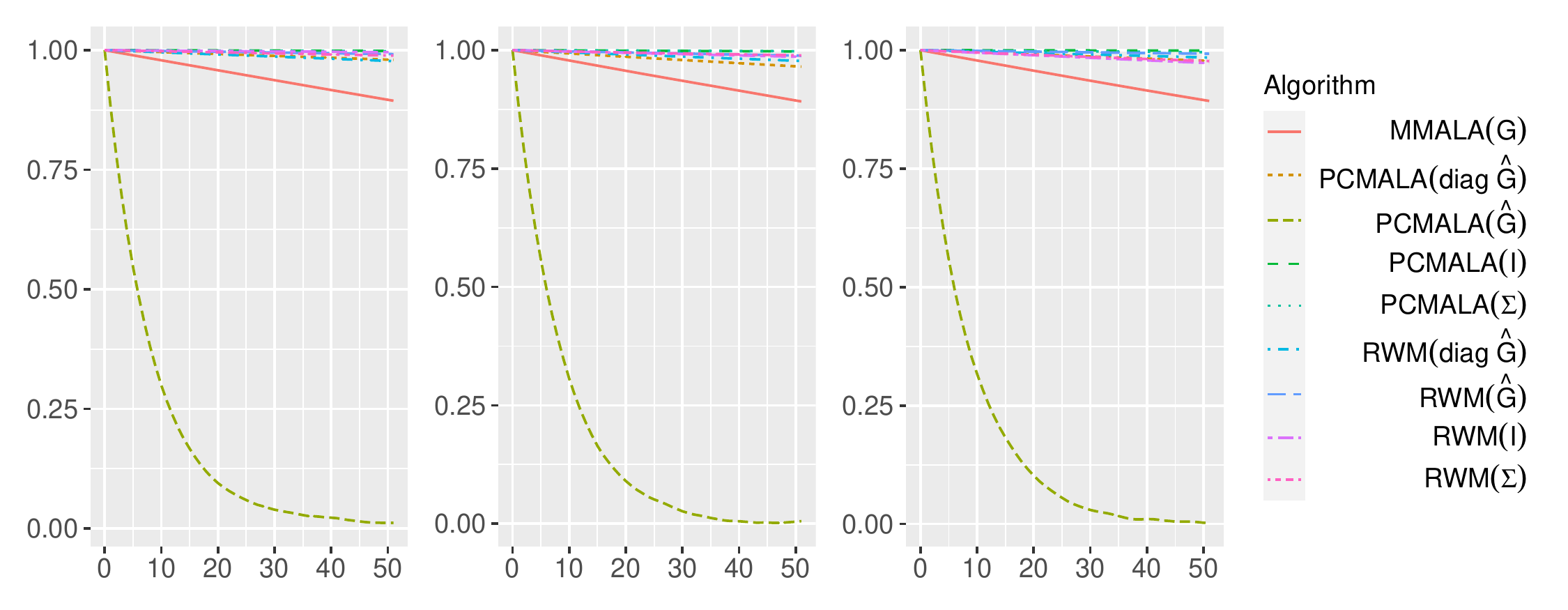}
  \caption{ACF plots for $x^{(1)}$ (left panel), $x^{(175)}$ (center
    panel), and $x^{(350)}$ (right panel) for the MH chains for the Poisson SGLMM with the log
    link. In the legend, $G$ refers to $\mathscr{I}^{-1}$ and $\hat{G}$ refers to $\hat{\mathscr{I}}^{-1}$.}
\label{fig:acf.poissonlog}
\end{figure*}

\begin{figure*}[b]
    \begin{minipage}[b]{0.24\linewidth}
    \includegraphics[width=\linewidth]{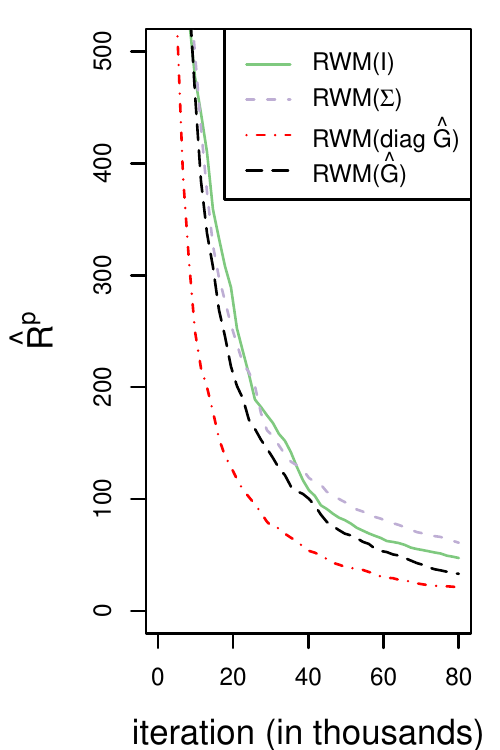}
  \end{minipage}
  \begin{minipage}[b]{0.24\linewidth}
    \includegraphics[width=\linewidth]{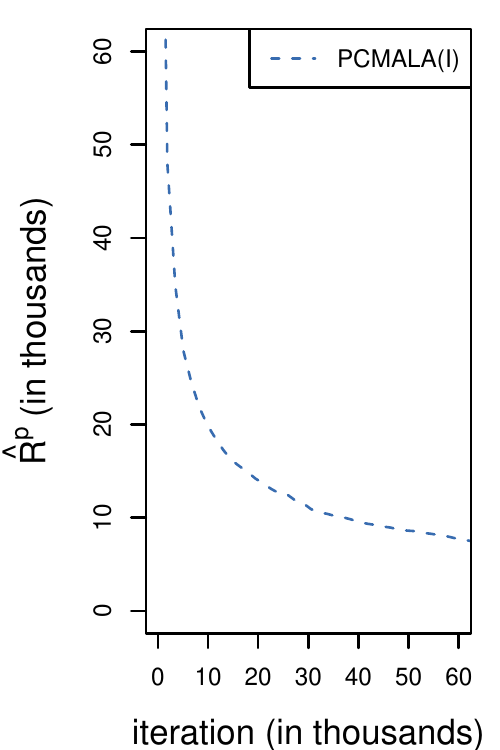}
  \end{minipage}
  \begin{minipage}[b]{0.24\linewidth}
    \includegraphics[width=\linewidth]{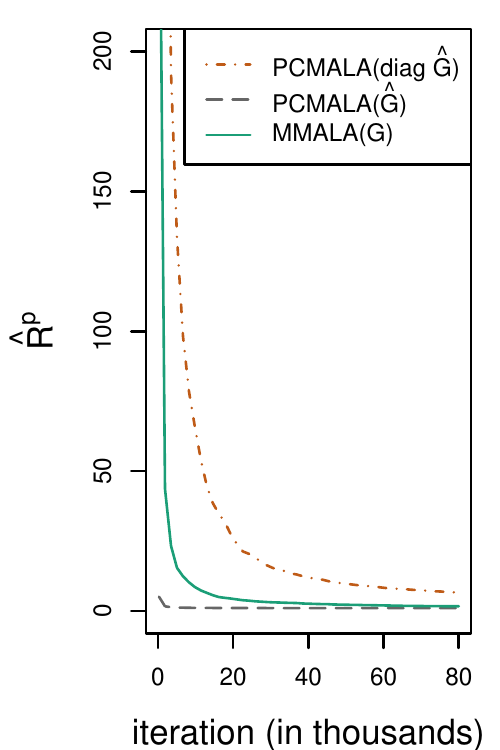}
  \end{minipage}
  \begin{minipage}[t]{0.24\linewidth}
    \includegraphics[width=\linewidth]{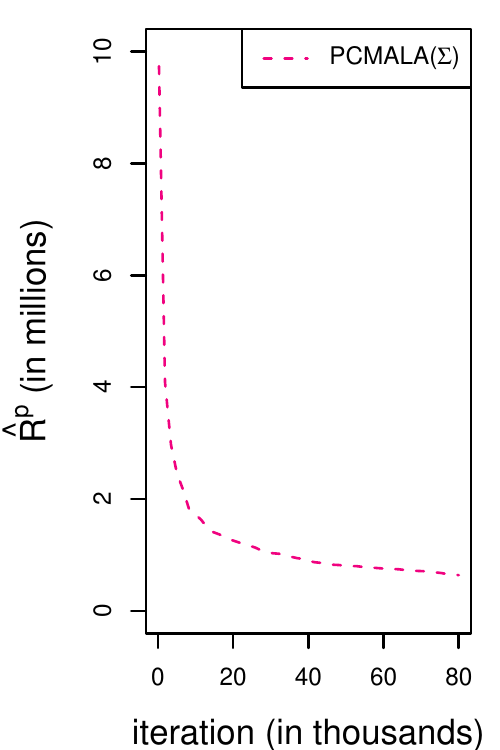}
  \end{minipage}
 \caption{Gelman and Rubin's $\hat{R}_p$ plot from the five parallel MH chains for the Poisson SGLMM with the log link. In the legend, $G$ refers to $\mathscr{I}^{-1}$ and $\hat{G}$ refers to $\hat{\mathscr{I}}^{-1}$.}
\label{fig:gelmandiag.poissonlog}
\end{figure*}

\begin{table}[h]
\caption{MSJD values for the PCULA chains for the binomial and Poisson SGLMMs}
\centering
\begin{tabular}{c c c c c| c c c c}
  \hline\hline
    &\multicolumn{4}{c}{binomial}& \multicolumn{4}{c}{Poisson}\\
  \hline
$G$ matrix &$I$ &$\Sigma$ &diag $(\hat{\mathscr{I}}^{-1})$ & $\hat{\mathscr{I}}^{-1}$& $I$ &$\Sigma$ &diag $(\hat{\mathscr{I}}^{-1})$ & $\hat{\mathscr{I}}^{-1}$ \\  
 \hline
 & 4.18&0.07 &7.42&7.41&0.05&0.05 &0.05&126.74\\
\hline\hline
\end{tabular}
\label{table:msejd.PCULA.bl}
\end{table}

\begin{figure*}[h]
  \includegraphics[width=\linewidth]{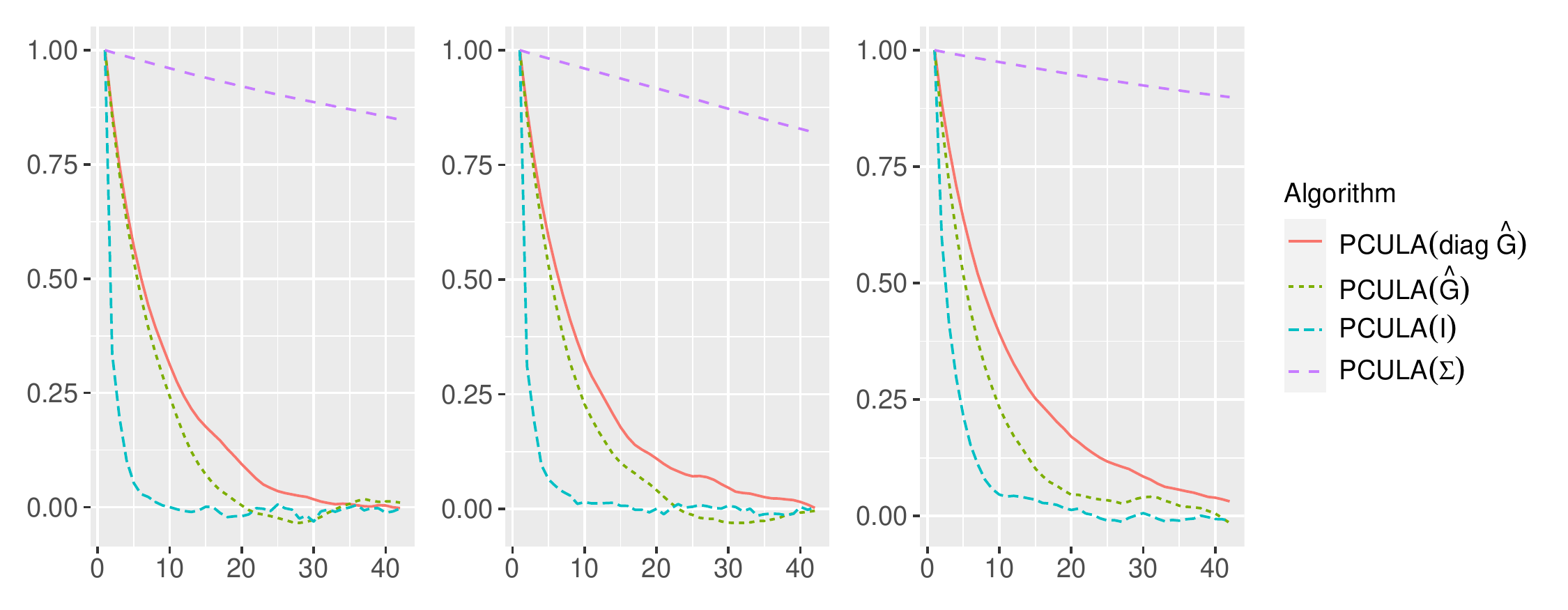}
  \caption{ACF plots for $x^{(1)}$ (left panel), $x^{(175)}$ (center
    panel), and $x^{(350)}$ (right panel) for the PCULA chains for the binomial SGLMM with the logit
    link. In the legend, $\hat{G}$ refers to $\hat{\mathscr{I}}^{-1}$.}
\label{fig:acf.PCULA.binomlogit}
\end{figure*}

\begin{figure*}[h]
  \includegraphics[width=\linewidth]{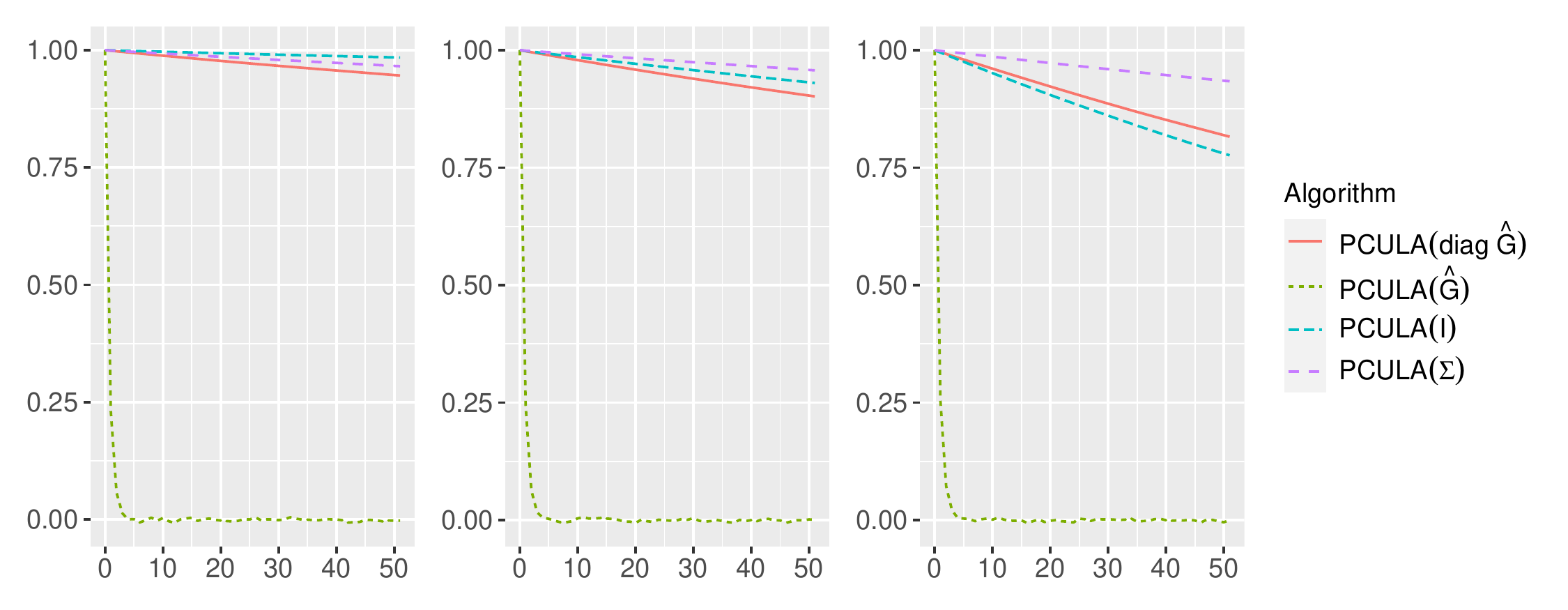}
  \caption{ACF plots for $x^{(1)}$ (left panel), $x^{(175)}$ (center
    panel), and $x^{(350)}$ (right panel) for the PCULA chains for the Poisson SGLMM with the log
    link. In the legend, $\hat{G}$ refers to $\hat{\mathscr{I}}^{-1}$.}
\label{fig:acf.PCULA.poissonlog}
\end{figure*}

\begin{figure}[h]
  \begin{center}
      \includegraphics[width=\linewidth]{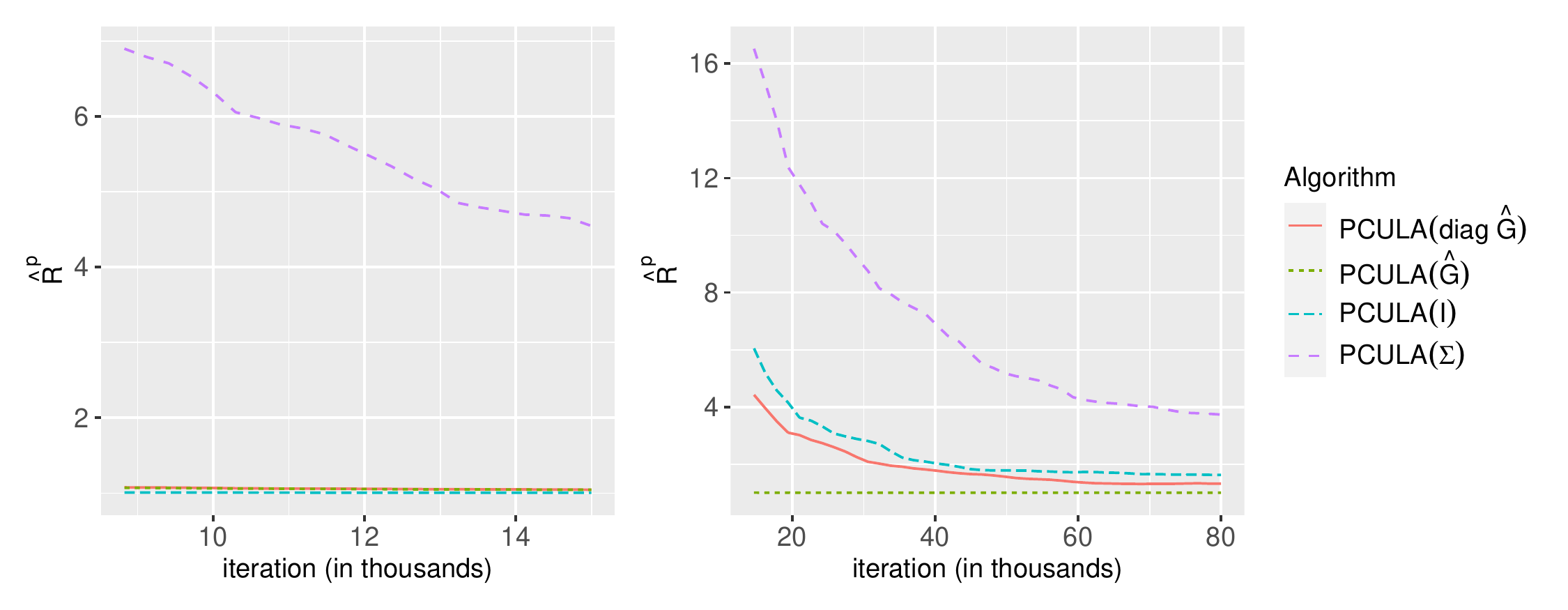}
  \end{center}
  \caption{Gelman and Rubin's $\hat{R}_p$ plot from the five parallel PCULA chains for the binomial-logit (left) and the Poisson-log SGLMMs (right). In the legend, $\hat{G}$ refers to $\hat{\mathscr{I}}^{-1}$.}
 \label{fig:gelmandiag.pcula}
 \end{figure}